\documentclass[aip,cha,reprint,floatfix]{revtex4-2}
\usepackage{graphicx}
\usepackage{dcolumn}
\usepackage{bm}
\usepackage[utf8]{inputenc}
\usepackage[T1]{fontenc}
\usepackage{algpseudocode}
\usepackage{mathptmx}
\usepackage{amsthm}
\usepackage{enumerate}
\usepackage{booktabs}
\usepackage{tabularx}
\usepackage[mathscr]{eucal}
\usepackage{amsmath,amssymb}
\usepackage{mathtools}
\usepackage{subfigure}
\usepackage[english]{babel}
\usepackage{lmodern}
\newtheorem{theorem}{Theorem}

\newtheorem{remark}{Remark}

\newcounter{algorithm}

\newcommand{\TT}{\mathbb T}
\newcommand{\RR}{\mathbb R}
\newcommand{\ZZ}{\mathbb Z}
\newcommand{\eps}{\varepsilon}

\newcommand\cO{{\mathcal O}}

\newcommand\mvector{\boldsymbol}

\newcommand\vv{\mvector{v}}

\newcommand\vx{\mvector{x}}
\newcommand\vy{\mvector{y}}

\newcommand\vA{\mvector{A}}

\newcommand\vF{\mvector{F}}

\newcommand\vI{\mvector{I}}

\newcommand\vK{\mvector{K}}

\newcommand\vM{\mvector{M}}

\newcommand\vR{\mvector{R}}

\newcommand\vX{\mvector{X}}
\newcommand\vY{\mvector{Y}}

\newcommand\vvarphi{\mvector{\varphi}}

\newcommand\vPhi{\mvector{\Phi}}

\newcommand\field{\mathbb}

\newcommand\T{\field{T}}

\newcommand\rmd{\mathrm{d}}

\newcommand\rmi{\mathrm{i}\mspace{1mu}}

\begin{document}
\mathtoolsset{
mathic 
}
\title{document}
\title{Dynamics, periodic orbits and \(C^1\) non-integrability  of the ABC flow}
\author{Wojciech Szumi{\'n}ski}
\email{w.szuminski@if.uz.zgora.pl}
\affiliation{Institute of Physics, University of Zielona G\'ora,
Licealna 9, PL-65-407 Zielona G\'ora, Poland}

\author{Jaume Llibre}
\email{jaume.llibre@uab.cat}
\affiliation{Departament de Matem\`atiques,
Universitat Aut\`onoma de Barcelona,
08193 Bellaterra, Barcelona, Catalonia, Spain}
\affiliation{Reial Acad\`emia de Ci\`encies i Arts de Barcelona,
La Rambla 115, 08002 Barcelona, Catalonia, Spain}
\begin{abstract}
We study the Arnold–Beltrami–Childress (ABC) flow in perturbative regimes near its three elementary integrable coordinate axes. Starting from a degenerate family of periodic streamlines of the limiting integrable system, we use first-order averaging to prove the existence of two isolated periodic solutions of the perturbed flow. The three perturbative regimes are related by the cyclic symmetry of the ABC vector field. In each case, the corresponding two-dimensional averaged system possesses two simple equilibria, one elliptic and the other hyperbolic, depending on the sign of the relevant parameter ratio. Poincaré sections illustrate how the degenerate periodic structure of the integrable limit breaks under perturbation and how the isolated periodic orbits predicted by averaging emerge within the surrounding dynamics. These periodic orbits provide the natural link between the perturbative analysis and the integrability problem. The eigenvalues of the linearized averaged system determine the leading-order behavior of the nontrivial characteristic multipliers of the bifurcating periodic solutions. Combining this relation with the Poincar\'e--Llibre--Valls
criterion, we prove that, in a neighborhood of each of these periodic orbits, there exists no nonconstant first integral \(H\in C^1\) that is regular along the orbit. The analytical results are further illustrated and confirmed by direct shooting, monodromy-matrix computations, and continuation of the elliptic and hyperbolic branches, demonstrating that the first-order averaging approximation remains quantitatively accurate over a substantially wider parameter range than the strict asymptotic regime.
\end{abstract}
\maketitle
\noindent\textbf{The Arnold--Beltrami--Childress (ABC) flow is a classical model of
three-dimensional incompressible dynamics in which regular and chaotic
motions coexist. Near its integrable coordinate axes, the flow approaches
a degenerate limit containing continuous families of periodic
streamlines. We show how these families break under small perturbations
for which all three ABC parameters are nonzero, and how isolated elliptic
and hyperbolic periodic orbits emerge from them. First-order averaging
proves the existence of these orbits, determines their leading-order
locations, and provides leading-order information on their
characteristic multipliers. This connects the breakup of the integrable
periodic structure directly with the non-existence of regular \(C^1\)
first integrals near the bifurcating orbits. Poincar\'e sections, shooting
computations, monodromy-matrix calculations, and numerical continuation
reveal how the resulting periodic branches are embedded in the
surrounding dynamics and how far the perturbative predictions remain
accurate beyond the asymptotic regime. The work thus provides a unified
perturbative description connecting degenerate integrable motion, the
emergence and stability of isolated periodic orbits, and \(C^1\)
non-integrability in the ABC flow.}
\section{Introduction}

We consider the Arnold--Beltrami--Childress vector field, commonly referred to as the   \textit{ABC flow}, which exhibits a remarkably wide variety of dynamical phenomena~\cite{Arnold1965,Beltrami1889,Childress1970}. It is
defined by the system $\dot \vx=\vv$ given explicitly by
\begin{equation}
\label{eq:abc}
\begin{cases}
\dot x = a\sin z+c\cos y,\\
\dot y = b\sin x+a\cos z,\\
\dot z = c\sin y+b\cos x.
\end{cases}
\end{equation}
where
\[
(x,y,z)\in\TT^3=(\RR/2\pi\ZZ)^3,
\qquad
a,b,c\in\RR.
\]
 
System~\eqref{eq:abc} is one of the classical models in the theory of
nonlinear dynamical systems. It describes a steady solution of the
Euler equations for an ideal incompressible fluid with periodic boundary
conditions. For this reason, the ABC flow has played an important role in
fluid mechanics, magnetohydrodynamics, and the theory of chaotic advection~\cite{FriedlanderYudovich1999}. In
particular, it provides a simple, explicit model in which periodic and quasi-periodic motions may coexist with chaotic trajectories~\cite{DombreEtAl1986}.

Further analytical results have clarified several specific features of
the ABC dynamics. In the paper~\cite{HuangZhaoDai1998}, Huang, Zhao, and Dai studied the near-integrable regime with one small coefficient and established conditions for the persistence
of invariant tori, while also deriving a Melnikov-type criterion for the
existence of chaotic streamlines. The existence
and location of stationary points for particular parameter families were
investigated by Ershkov~\cite{Ershkov2016}. More recently, Xin, Yu, and
Zlato\v{s} in~\cite{XinYuZlatos2016} proved the existence of nonperturbative periodic orbits in the
fully symmetric case \(a=b=c=1\), with rotation vectors parallel to the
three coordinate directions. These results illustrate the range of analytical approaches that have
been used to study invariant structures and periodic dynamics in the ABC
flow.

The  ABC system~\eqref{eq:abc} is a volume-preserving flow, which is particularly relevant
to the integrability problem considered in this paper. Indeed, the vector
field \(\vv=(v_1,v_2,v_3)^{\mathsf T}\) associated with
system~\eqref{eq:abc} is divergence-free:
\[
\nabla\cdot \vv
=
\frac{\partial v_1}{\partial x}
+
\frac{\partial v_2}{\partial y}
+
\frac{\partial v_3}{\partial z}
=
0.
\]
Consequently, the flow preserves the standard volume form on \(\TT^3\).
In particular, if \(\Gamma\) is a periodic orbit and \(M_\Gamma\)
denotes the monodromy matrix of the variational equations along
\(\Gamma\), then Liouville's formula gives
\[
\det M_\Gamma=1.
\]
This identity is useful in the numerical computations, since it provides
a direct consistency check for the integration of the variational
equations, which will be computed in the second part of the paper. The analytical non-integrability argument developed in this
paper, however, does not rely on this determinant relation alone. It is
based on the characteristic multipliers of the periodic orbits and on the
Poincar\'e--Llibre--Valls criterion for the non-existence of a regular
\(C^1\) first integral~\cite{Poincare1891a,Poincare1897,Llibre:11a::}.

The integrability problem for the ABC flow has a long and engaging history. Since the system is three-dimensional and possesses the invariant
volume measure, the existence of one nonconstant first integral,
together with the Jacobi last multiplier, is sufficient to reduce the
system to quadratures~\cite{Whittaker1988}. Thus, the crucial question is whether there exists a nonconstant scalar function
\[
 F(x,y,z):\TT^3\longrightarrow\RR
 \]
such that
\[
\frac{\rmd}{\rmd t}F(x(t),y(t),z(t))=0
 \]
along every solution $(x(t),y(t),z(t))$ of~\eqref{eq:abc}. If such a first integral exists, then
the streamlines governed by~\eqref{eq:abc} are confined to the invariant surface 
\(F=\mathrm{const}\). In its absence, the dynamics is not restricted by
such a global foliation and may exhibit genuinely three-dimensional
transport.

The ABC flow possesses elementary integrable subfamilies when one of the
parameters \(a,b,c\) vanishes. These subfamilies form the coordinate
planes in the parameter space, and their pairwise intersections define the
integrable axes. The present paper studies small perturbations of these
axes in directions transverse to the corresponding integrable coordinate
planes. Perturbing only one of the two transverse parameters keeps the
system inside an integrable coordinate plane, whereas perturbing both
produces a genuinely three-dimensional perturbation of the degenerate
integrable flow. 

Rigorous results on the non-integrability of the ABC flow are considerably
more subtle than the elementary integrable reductions described above.
Ziglin~\cite{Ziglin1996,Ziglin1998} established partial non-integrability results by studying the monodromy of the normal
variational equations along a certain particular solution. In
particular, he proved the absence of real meromorphic first integrals in
several parameter regimes, including cases with two equal coefficients,
up to certain exceptional parameter values. The fully symmetric case
\(a^2=b^2=c^2\neq0\), which is one of the most extensively studied cases
from the numerical point of view, was also shown to be non-integrable.
Later on, Maciejewski and Przybylska~\cite{MaciejewskiPrzybylska2002}
strengthened and completed part of this analysis by combining Ziglin's
theory with  the differential Galois approach. They proved, in particular, the
meromorphic non-integrability of the ABC flow in the case
\(a^2=b^2\) and \(abc\neq0\).

Llibre and Valls~\cite{LlibreValls2012} used a different approach, deriving sufficient conditions for the non-existence of $C^1$ first integrals for the ABC flow. Their analysis relied on a particular periodic solution previously obtained by Chicone through geometric regular perturbation theory~\cite{Chicone1995}. Using this known periodic orbit, they computed the associated monodromy matrix and applied a Poincar\'e-type criterion to establish $C^1$ non-integrability under an appropriate transversality condition.

The problem considered in this paper is different. We study periodic orbits bifurcating from the degenerate periodic families associated with the three integrable axes of the ABC flow. The Poincar\'e sections show what happens to these families under perturbation. In the integrable limit, the periodic orbits form continuous families, while after both parameters are switched on, this degeneracy is broken and isolated periodic orbits appear. We use first-order averaging to study this phenomenon analytically and to determine which periodic orbits bifurcate from the unperturbed families. Near the a-axis, after a transverse blow-up and the introduction of the fast angular variable, the problem reduces to a two-dimensional averaged system with two simple equilibria. They give rise to two periodic solutions of the ABC flow. The cases of the b- and c-axes are obtained by cyclic symmetry. 

The averaging analysis also provides the spectral information needed for the integrability analysis. By the theorem of Llibre and Zhang~\cite{LlibreZhang:R3}, the eigenvalues of the Jacobian matrix of the averaged vector field determine, to first order, the nontrivial characteristic multipliers of the bifurcating periodic solutions. Hence, once the averaged equilibria are shown to be simple, the corresponding transverse multipliers are different from 1 for all sufficiently small $\varepsilon>0$. The only remaining multiplier equal to 1 is the trivial one associated with the flow direction. The Poincar\'e–Llibre–Valls criterion can therefore be applied to exclude the existence of a nonconstant $C^1$ first integral whose gradient does not vanish along the periodic orbit.

The contribution of the present paper is thus twofold. First, it provides
an explicit averaging construction of periodic streamlines bifurcating
from the degenerate periodic families located on the three integrable
axes. Second, it combines this construction with the
Llibre--Zhang relation between averaged eigenvalues and characteristic
multipliers in order to obtain a direct  criterion for regular
\(C^1\)-non-integrability. In this way, the proof does not require a
separate perturbative computation of the monodromy matrix of the periodic
orbit: both the existence of the orbit and the leading-order spectral
obstruction are extracted from the same averaged system.

The analytical results are complemented by a direct numerical study of the original ABC system. Poincar\'e sections are used to visualize the breakup of the degenerate periodic families as the system is perturbed away from the integrable axes and the isolated periodic orbits predicted by averaging emerge. The bifurcating periodic solutions are then refined by a shooting method, and their characteristic multipliers are computed from the corresponding monodromy matrices. These numerical values are compared with the first-order predictions obtained from the averaged system. The computations confirm the asymptotic formulas with high accuracy in the perturbative regime and show that the qualitative elliptic and hyperbolic classification remains valid over a wider parameter range. Numerical continuation is further used to follow both periodic branches away from the perturbative regime and to determine how the accuracy of the averaged approximation deteriorates with increasing distance from the integrable axis. The numerical analysis therefore follows specifically the periodic branches constructed by averaging and quantifies the range over which
the corresponding asymptotic predictions remain accurate.

\section{The ABC flow and numerical evidence}
\label{sec:abc-numerical-evidence}

In this section we describe the elementary integrable subfamilies of the ABC
flow and explain the perturbative setting used in the rest of the paper. We
then present Poincar\'e sections near the integrable \(a\)-axis, which provide
the numerical motivation for the averaging construction.

\subsection{The ABC vector field and integrable axes}
\label{subsec:abc-integrable-axes}

We recall the elementary integrable subfamilies of the ABC flow. If
\(c=0\), system~\eqref{eq:abc} becomes
\[
\dot x=a\sin z,\qquad
\dot y=b\sin x+a\cos z,\qquad
\dot z=b\cos x .
\]
In this case, the variables \((x,z)\) form a closed two-dimensional
subsystem, and
\[
H_{c=0}(x,z)=a\cos z+b\sin x
\]
is a first integral. Indeed,
\[
\frac{\rmd}{\rmd t}H_{c=0}
=
-a\sin z\,\dot z+b\cos x\,\dot x=0.
\]
The remaining variable \(y\) is then obtained by quadrature.

The other two integrable cases follow from the cyclic symmetry
\begin{equation}
\label{eq:sym}
(x,y,z)\mapsto(y,z,x),
\qquad
(a,b,c)\mapsto(b,c,a),
\end{equation}
 with the corresponding first integrals obtained by applying the same
permutation to \(H_{c=0}\).    

The \(a\)-axis is the intersection of the two integrable coordinate planes
\(b=0\) and \(c=0\). We denote it by
\begin{equation}
\label{eq:aa}
\mathcal A_a=\{(a,b,c):a\neq0,\ b=0,\ c=0\}.
\end{equation}
On this axis system~\eqref{eq:abc} reduces to
\begin{equation*}
\label{eq:vh0}
\dot x=a\sin z,\qquad
\dot y=a\cos z,\qquad
\dot z=0.
\end{equation*}
In particular, the torus \(z=0\) is invariant and is foliated by the periodic
solutions
\begin{equation}
\label{eq:sol}
x(t)=x_0,\qquad y(t)=y_0+at,\qquad z(t)=0,
\end{equation}
with unperturbed period
\[
T_0=\frac{2\pi}{|a|}.
\]
Thus the unperturbed flow on \(\mathcal A_a\) possesses a degenerate
one-parameter family of periodic solutions.
For this reason, in the analytical part of the paper we perform the complete
averaging calculation  near \(\mathcal A_a\); the two remaining regimes
can be  obtained by cyclic permutation~\eqref{eq:sym}.

\subsection{Poincar\'e sections near the integrable \(a\)-axis}
\label{subsec:poincare-a-axis}

We now provide numerical evidence for the transition from strongly chaotic dynamics to the integrable $a$-axis. Throughout this subsection, we fix
$
a=1
$
and move to the \(a\)-axis by gradually decreasing values of the remaining parameters
\[
(b,c)=(\rho,\rho),\qquad \rho\downarrow0.
\]
This choice is made solely for visualization. It provides a simple view of the system dynamics approaching the integrable regime, with $\rho= 0$.

The Poincar\'e section is defined by
\[
y=0\pmod{2\pi},\qquad \dot y>0,
\]
and is projected onto the \((x,z)\)-plane. For visualization purposes, the
angular coordinate \(x\) is represented in the fundamental interval
\([-\pi,\pi]\), while the range of \(z\) is restricted to the region of
interest. Thus, we plot
\[
(x,z)\in[-\pi,\pi]\times I,
\]
where the vertical interval \(I\) is chosen separately in each panel to make
the relevant structures clearly visible.

Fig.\ref{fig:abc-poincare-transition} shows the resulting sections for
\[
b=c=10^{-1},\quad b=c=10^{-2},\quad b=c=10^{-3},
\quad b=c=0.
\]
Each color represents the intersections generated by a distinct initial
condition. For \(b=c=10^{-1}\), shown in
Fig.~\ref{fig:abc-poincare-transition}(a), the section exhibits a broad
region of scattered points near \(x=-\pi/2\), consistent with chaotic
trajectories arising in the neighborhood of the broken separatrix. In
contrast, near \(x=\pi/2\), one observes a family of smooth invariant
curves surrounding a periodic orbit.

For smaller values of the parameters \(b=c=10^{-2}\), see
Fig.~\ref{fig:abc-poincare-transition}(b), the chaos contracts
significantly. Most of the section is now covered by regular invariant curves,
whereas the remaining chaotic trajectories are confined to a narrow layer
surrounding the separatrix-like structure issuing from the neighborhood of
\(x=-\pi/2\). At the same time, the elliptic island centred near \(x=\pi/2\)
remains unchanged.

For \(b=c=10^{-3}\), the Poincar\'e section shown in
Fig.~\ref{fig:abc-poincare-transition}(c) is almost entirely regular.
No sizable chaotic layer is visible at the scale of the figure, while
the separatrix-like structure near \(x=-\pi/2\) and the elliptic island
near \(x=\pi/2\) remain clearly identifiable. Thus, as \(b\) and \(c\)
decrease, the section becomes progressively more regular and approaches
the degenerate integrable structure observed in the limit \(b=c=0\).

\begin{figure*}[htp]
\centering
\subfigure[$b=c=10^{-1}$]{
\includegraphics[width=0.4\textwidth]{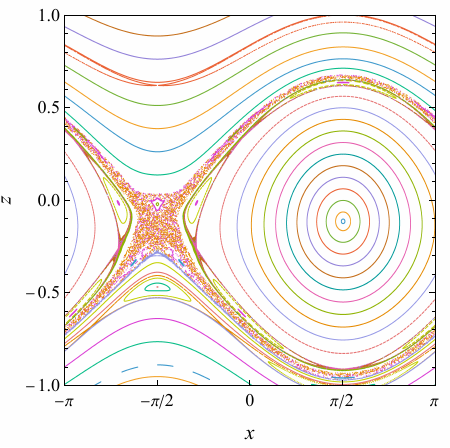}}\hspace{0.5cm}
\subfigure[$b=c=10^{-2}$]{
\includegraphics[width=0.4\textwidth]{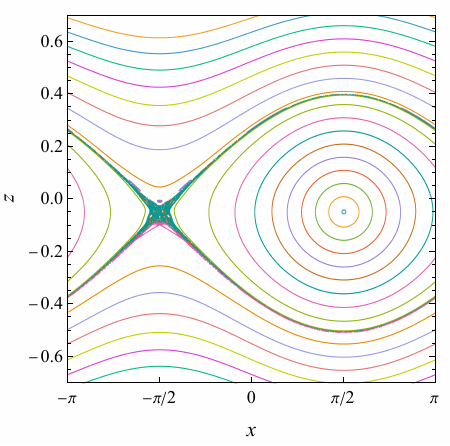}}\\
\subfigure[$b=c=10^{-3}$]{
\includegraphics[width=0.4\textwidth]{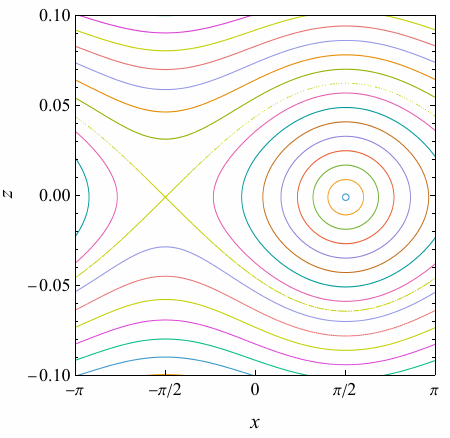}}\hspace{0.5cm}
\subfigure[$b=c=0$]{
\includegraphics[width=0.405\textwidth]{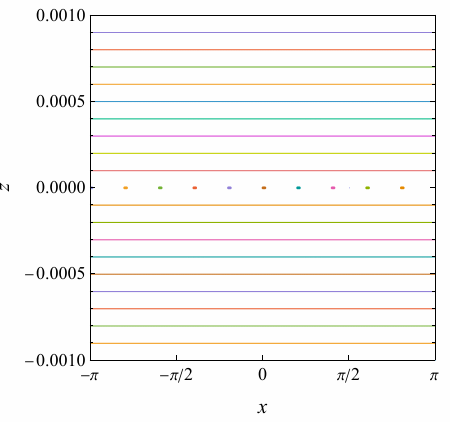}}
\caption{(Color online)
Poincar\'e sections of the ABC flow~\eqref{eq:abc} for \(a=1\) and
decreasing values of \(b=c\). The section is defined by
\(y=0\pmod{2\pi}\), \(\dot y>0\), and represented in the
\((x,z)\)-plane. Each color corresponds to a distinct initial condition. Panel~(a), \(b=c=10^{-1}\), shows a broad chaotic
layer near \(x=-\pi/2\). In panel~(b), \(b=c=10^{-2}\), this layer is
considerably reduced and concentrated near a separatrix-like structure.
For \(b=c=10^{-3}\), panel~(c) displays an almost integrable structure.
Panel~(d) shows the integrable limit \(b=c=0\), where the level \(z=0\)
represents the degenerate family of periodic solutions. The range of
the vertical coordinate \(z\) is adjusted separately in each panel in
order to resolve the shrinking structures.}
\label{fig:abc-poincare-transition}
\end{figure*}

Finally, Fig.~\ref{fig:abc-poincare-transition}(d) corresponds to the
integrable limit \(b=c=0\). As explained in
Subsection~\ref{subsec:abc-integrable-axes}, the invariant torus \(z=0\) is
foliated by the degenerate family of periodic solutions~\eqref{eq:sol}.
Hence, in the Poincar\'e section \(y=0\pmod{2\pi}\), the line \(z=0\)
represents the trace of this one-parameter family for the integrable
\(a\)-axis system.

The numerical results suggest that, when both $b$ and $c$ are small
and nonzero, this degenerate periodic family breaks into two isolated periodic
orbits, one elliptic and one hyperbolic. This is precisely the situation
studied analytically in Section~\ref{sec:periodic-streamlines}, where the
existence of these two periodic solutions is proved by first-order
averaging. Their characteristic multipliers are subsequently used to
establish the $C^1$-non-integrability of the ABC flow.

\section{Preliminaries}
\label{sec:preliminaries}

In this section we recall the main tools used in the proof of our
results. The first is the classical first-order averaging theorem,
which gives periodic solutions from simple zeros of the averaged vector
field. The second is the Poincar\'e--Llibre--Valls criterion, which
provides an obstruction to the existence of a regular \(C^1\) first
integral in terms of the characteristic multipliers of a periodic
orbit. We also recall the relation between the eigenvalues of the
averaged system and the characteristic multipliers of the periodic
solutions obtained by averaging. For a more comprehensive exposition of these topics, we refer the reader to~\cite{Llibre:11a::,LlibreZhang:R3,LLIBRE2023103943}.

\subsection{ The first-order averaging}
\label{subsec:first-order-averaging}

Averaging theory is a well--known perturbative approach in the qualitative theory of
differential equations. In its simplest form, it replaces a non-autonomous
system with a small periodic perturbation by an autonomous averaged system
\cite{Verhulst1996,LlibreNovaesTeixeira:14,GineGrauLlibre:13,
BuicaLlibre:04,LlibreZhang:R3,LlibreTian:R4}. The simple equilibria of this
averaged system generate periodic solutions of the original system.
In the present paper, this method will be applied near the integrable axes of
the ABC flow, after one of the angular variables has been chosen as the new
independent variable.

 Consider the system
\begin{equation}
\label{eq:averaging-standard}
\frac{\rmd \vX}{\rmd s}
=
\varepsilon\,\vF(s,\vX)
+
\varepsilon^2\,\vR(s,\vX,\varepsilon),
\qquad
\vX\in\Omega\subset\RR^m,
\end{equation}
where \(\Omega\) is an open set, the functions \(\vF\) and \(\vR\) are
sufficiently smooth, and both are \(T\)-periodic with respect to the variable
\(s\). The first averaged vector field is defined by
\begin{equation}
\label{eq:averaged-function}
\bar{\vF}(\vX)
=
\frac{1}{T}
\int_0^T
\vF(s,\vX)\,\rmd s.
\end{equation}

\begin{theorem}[First-order averaging theorem]
\label{thm:first-order-averaging}
Assume that \(\vX_*\in\Omega\) is a simple zero of the averaged vector field
\(\bar{\vF}\), that is,
\begin{equation}
\label{eq:simple-zero}
\bar{\vF}(\vX_*)=0,
\qquad
\det D_{\vX}\bar{\vF}(\vX_*)\neq0.
\end{equation}
Then, for all sufficiently small \(\varepsilon>0\), system
\eqref{eq:averaging-standard} possesses a \(T\)-periodic solution
\(\vX(s,\varepsilon)\) such that
\[
\vX(0,\varepsilon)\longrightarrow \vX_*
\qquad
\text{as}
\qquad
\varepsilon\to0.
\]
Equivalently, every simple zero of the averaged vector field gives rise to a
periodic solution of the original perturbed system.
\end{theorem}
The following remark explains how the abstract notation in
Theorem~\ref{thm:first-order-averaging} will be used in the three cyclic
regimes of the ABC flow.

\begin{remark}\upshape
In the applications to the ABC flow, the independent variable \(s\) is
one of the angular variables \(x,y,z\). Hence, in all three cyclic regimes
considered below, the averaging period is
$
T=2\pi.
$
Consistently with the angular convention adopted in the Introduction, we
shall use the symmetric averaging interval
$
[-\pi,\pi].
$ 
Since the integrands are \(2\pi\)-periodic, this choice is equivalent to
the interval \([0,2\pi]\).
More precisely, near the \(a\)-axis we take \(s=y\), near the \(b\)-axis we take
\(s=z\), and near the \(c\)-axis we take \(s=x\).

The variable \(\vX\) in
Theorem~\ref{thm:first-order-averaging} will be identified with the
corresponding reduced variables
\[
\vx_a=(x,Z)^{\mathsf T},\qquad
\vx_b=(y,X)^{\mathsf T},\qquad
\vx_c=(z,Y)^{\mathsf T}
\]
in the three cyclic regimes, respectively. The detailed derivation of these
reduced systems is given in Section~\ref{sec:periodic-streamlines}.
\end{remark}
\subsection{Characteristic multipliers and  \(C^1\)-non-integrability}
\label{subsec:characteristic-multipliers}
The problem of integrability is one of the classical questions in the
qualitative theory of differential equations and mathematical physics. There
are several powerful methods for proving non-integrability, including the
Kowalevskaya--Painlev\'e analysis~\cite{Painleve:1902,Goriely:2001}, Ziglin's
theory~\cite{Ziglin:82::b}, and the Morales--Ramis theory based on the differential Galois approach~\cite{Morales:99::,Morales:01::}. These methods have been
successfully applied to many systems arising in mechanics and mathematical
physics~
\cite{Maciejewski:11::,Maciejewski:05::,10.1063/5.0200592,
mp:13::e,Combot:18::,Szuminski2024,Szuminski2025,Szuminski2026}.

Here we follow a different approach. Instead of using complex-analytic
or differential Galois techniques, we use the classical idea of
Poincar\'e~\cite{Poincare1891a,Poincare1897}, which relates the existence of
first integrals to the characteristic multipliers of periodic orbits. In the
form developed by Llibre and Valls~\cite{Llibre:11a::}, this idea gives a
criterion for the non-existence of regular   \(C^1\) first integrals.

In the present setting this criterion is particularly convenient because the
periodic orbits are obtained by first-order averaging. We shall use the
relation between the eigenvalues of the averaged vector field and the
nontrivial characteristic multipliers of the corresponding periodic orbits,
as established in~\cite{LlibreZhang:R3}. Thus the  
\(C^1\)-non-integrability problem is reduced to the  analysis of the
Jacobian matrices of the averaged systems.

We now recall the   criterion for the non-existence of \(C^1\) first
integrals based on characteristic multipliers of periodic orbits.  Following~\cite{Llibre:11a::}, consider an autonomous differential system
\begin{equation}
\label{eq:autonomous-general}
\dot{\vx}=\vv(\vx),
\qquad
\vx\in\mathcal U\subset\RR^n,
\end{equation}
where \(\vv\in C^1(\mathcal U)\). Assume that
\(\vvarphi(t,\vx_0)\) is a periodic solution of minimal period \(T>0\), that is,
\[
\vvarphi(T,\vx_0)=\vx_0,
\qquad
\vvarphi(t,\vx_0)\neq \vx_0,
\qquad
0<t<T.
\]
The corresponding periodic orbit is
\[
\Gamma
=
\{
\vvarphi(t,\vx_0):0\le t\le T
\}.
\]
The variational equations along this periodic solution are as follows
\begin{equation}
\label{eq:var-general}
\dot{\vy}
=
\vA(t)\vy,
\qquad
\vA(t)=D_{\vx}\vv(\vvarphi(t,\vx_0)),
\end{equation}
where $\vA(t)$ is the Jacobian matrix of the vector field $\vv$ at $\vvarphi(t,\vx_0)$. 

Let \(\vPhi(t)\) be the fundamental matrix solution of
\eqref{eq:var-general} satisfying
\[
\vPhi(0)=\vI_n.
\]
The monodromy matrix associated with the periodic orbit \(\Gamma\) is
\begin{equation}
\label{eq:monodromy-general}
\vM_\Gamma=\vPhi(T).
\end{equation}
Its eigenvalues
\[
\mu_1,\ldots,\mu_n
\]
are called the \textit{characteristic multipliers}, or \textit{Floquet multipliers}, of
\(\Gamma\). Since system~\eqref{eq:autonomous-general} is autonomous, one
multiplier is always equal to \(1\). This trivial multiplier corresponds to
the tangent direction to the periodic orbit and reflects the invariance of
autonomous systems under time translation. For further details, see Proposition~1 in~\cite{Llibre:11a::}.

The following criterion, due to Poincar\'e in its classical form and used here
in the version developed by Llibre and Valls, provides the obstruction to
\(C^1\)-integrability employed in this paper.

\begin{theorem}[Poincar\'e--Llibre--Valls criterion]
\label{thm:PLV}
Let \(\Gamma\) be a periodic orbit of system~\eqref{eq:autonomous-general}.
Assume that the characteristic multiplier \(1\) has algebraic multiplicity
one. Then system~\eqref{eq:autonomous-general} admits no nonconstant first
integral \(H\) of class \(C^1\), defined in a neighborhood of \(\Gamma\), such
that
\[
\nabla H(\vx)\neq0,
\qquad
\vx\in\Gamma.
\]
\end{theorem}

The condition \(\nabla H(\vx)\neq0\) means that the first integral is regular
along the periodic orbit. Indeed, if \(H\) is a first integral, then the orbit
\(\Gamma\) is contained in a level surface \(H(\vx)=h\), and the vector
\(\nabla H(\vx)\) is normal to this surface. Therefore, the hypothesis excludes
the degenerate situation in which the periodic orbit lies entirely in the
critical set of the first integral.

\subsection{Relation between the eigenvalues of the averaged system and the characteristic multipliers}

The direct application of Theorem~\ref{thm:PLV} requires the computation
of the characteristic multipliers of a periodic orbit, that is, the
eigenvalues of the corresponding monodromy matrix. In general, this requires solving the variational equations along the periodic solution and evaluating the fundamental matrix over one period,
which may be difficult for systems depending on parameters.

For periodic solutions obtained by first-order averaging, the
leading-order behavior of the characteristic multipliers can instead
be inferred from the linearization of the averaged vector field. We use
the relation formulated by Llibre and Zhang~\cite{LlibreZhang:R3}; a
closely related argument also appears in the study of four-dimensional
zero--Hopf bifurcations by Llibre and Tian~\cite{LlibreTian:R4}.
Since the result is central to our non-integrability argument and the
Llibre--Zhang reference is currently available only as a preprint, we
include in Appendix~\ref{app:LZ-proof} a self-contained proof adapted
to the present notation.

\begin{theorem}[Llibre--Zhang]
\label{thm:LZ}
Consider system~\eqref{eq:averaging-standard} under the assumptions of
Theorem~\ref{thm:first-order-averaging}, and let
\(\vX_*\in\Omega\) be a simple zero of the averaged vector field
\(\bar{\vF}\). Assume, in addition, that \(\vF\) is of class \(C^2\)
with respect to \(\vX\) and that \(\vR\) is of class \(C^1\) with
respect to \(\vX\), uniformly on the compact subsets under
consideration.

Let \(\vX(s,\eps)\) be the \(T\)-periodic solution bifurcating from
\(\vX_*\). Then
\begin{equation*}
\label{eq:periodic-solution-first-order}
\vX(s,\eps)
=
\vX_*+\cO(\eps)
\end{equation*}
uniformly for \(s\in[0,T]\).

Let \(\vPhi(s,\eps)\) denote the fundamental matrix solution of the
variational system
\begin{equation}
\label{eq:var-averaged-system}
\begin{aligned}
&\frac{\rmd\vY}{\rmd s}
=
\eps
\left.
D_{\vX}
\bigl(
\vF(s,\vX)
+
\eps\vR(s,\vX,\eps)
\bigr)
\right|_{\vX=\vX(s,\eps)}
\vY,
\\
&\vPhi(0,\eps)
=
\vI_m.
\end{aligned}
\end{equation}
Then the corresponding monodromy matrix satisfies
\begin{equation}
\label{eq:LZ-monodromy}
\vPhi(T,\eps)
=
\exp\!\left(
\eps T D_{\vX}\bar{\vF}(\vX_*)
\right)
+
\cO(\eps^2).
\end{equation}

Let \(\lambda_1,\ldots,\lambda_m\) be the eigenvalues of
\(D_{\vX}\bar{\vF}(\vX_*)\). If these eigenvalues are simple, then,
after a suitable labeling, the characteristic multipliers
\(\rho_1(\eps),\ldots,\rho_m(\eps)\) of the reduced periodic system
satisfy
\begin{equation}
\label{eq:multiplier-averaging}
\rho_j(\eps)
=
\exp\!\left(
\eps T\lambda_j
\right)
+
\cO(\eps^2),
\qquad
j=1,\ldots,m.
\end{equation}
\end{theorem}

We now relate the multipliers of the reduced periodic system to those
of the corresponding periodic orbit of the original autonomous flow.
Suppose that the  averaged system is obtained from an
autonomous system of dimension \(n\), and that the independent
variable is an angular variable. Then
\[
m=n-1,
\qquad
T=2\pi.
\]
Assume, in addition, that the angular velocity does not vanish along
the corresponding periodic orbit. If it is positive, then increasing
the angular variable by \(2\pi\) defines the positive-time local
Poincar\'e return map. Consequently, the characteristic multipliers
\(\mu_1(\eps),\ldots,\mu_n(\eps)\) of the periodic orbit can be
labeled so that
\[
\mu_1(\eps)=1,
\qquad
\mu_{j+1}(\eps)=\rho_j(\eps),
\qquad
j=1,\ldots,m,
\]
where \(\rho_j(\eps)\) are the characteristic multipliers of the
reduced \(T\)-periodic averaged system.

If the angular velocity is negative, then the positive-time return
corresponds to decreasing the angular variable by \(2\pi\). The
corresponding Poincar\'e return map is therefore the inverse of the
map obtained by increasing the angular variable by \(2\pi\). Hence,
after a suitable relabeling,
\[
\mu_1(\eps)=1,
\qquad
\mu_{j+1}(\eps)=\rho_j(\eps)^{-1},
\qquad
j=1,\ldots,m.
\]

The above discussion applies to any autonomous system satisfying the
stated assumptions. 

For the ABC flow considered in this paper, one has~\(n=3\) and therefore \(m=2\). Moreover, the averaged Jacobian matrices
obtained below have eigenvalues of the form
$
\lambda,\,-\lambda.
$
Hence, at the level of the first-order asymptotics, replacing
\(\rho_j(\eps)\) by its inverse only interchanges the two nontrivial
characteristic multipliers. Therefore, independently of the orientation
of the angular variable and after a possible relabeling, they satisfy
\begin{equation}
\label{eq:full-multiplier-averaging-abc}
\mu_{j+1}(\eps)
=
\exp\!\left(
2\pi\eps\lambda_j
\right)
+
\cO(\eps^2),
\qquad
j=1,2.
\end{equation}

If \(\lambda_j\neq0\), then
\[
\mu_{j+1}(\eps)-1
=
2\pi\eps\lambda_j+\cO(\eps^2)
\neq0
\]
for all sufficiently small \(\eps>0\). Thus, if both eigenvalues of
the averaged Jacobian are nonzero, neither of the two nontrivial
characteristic multipliers is equal to \(1\). Since an autonomous
system always has the trivial multiplier \(\mu_1(\eps)=1\), it then
follows that \(1\) has algebraic multiplicity one.

Consequently, for the ABC flow, the application of the
Poincar\'e--Llibre--Valls criterion reduces to verifying that the
eigenvalues of the corresponding averaged Jacobian matrices are
nonzero. This will be done in
Section~\ref{sec:nonint}.
\section{Periodic solutions near the integrable axes}
\label{sec:periodic-streamlines}

In this section we prove the existence of periodic solutions of the ABC flow
near the three integrable axes. We give the complete calculation near the
\(a\)-axis \(\mathcal A_a\), while the other two cases follow from the cyclic
symmetry~\eqref{eq:sym}.

We consider a perturbation of \(\mathcal A_a\), with \(a\neq0\) fixed,
of the form
\[
b=\varepsilon^2 b_1,\qquad
c=\varepsilon^2 c_1,
\qquad
0<\varepsilon\ll1.
\]
When both \(b_1\) and \(c_1\) are nonzero, the perturbation leaves the
two integrable coordinate planes whose intersection defines
\(\mathcal A_a\). The averaging construction below, however, only
requires \(b_1\neq0\). We then apply the first-order averaging theorem
recalled in Section~\ref{sec:preliminaries}.
\subsection{Existence theorem}
\label{subsec:existence-theorem}

We first state the existence result in the three cyclic perturbative regimes.

\begin{theorem}
\label{thm:main-cyclic}
Consider the ABC flow~\eqref{eq:abc} on \(\TT^3\). Assume that one of the
following perturbative regimes holds:
\begin{align}
\mathrm{(A)}\qquad
&a\neq0,
&&b=\varepsilon^2 b_1,
&&c=\varepsilon^2 c_1,
&&b_1\neq0,
\notag\\[1mm]
\mathrm{(B)}\qquad
&b\neq0,
&&c=\varepsilon^2 c_1,
&&a=\varepsilon^2 a_1,
&&c_1\neq0,
\label{eq:cyclic-regimes}
\\[1mm]
\mathrm{(C)}\qquad
&c\neq0,
&&a=\varepsilon^2 a_1,
&&b=\varepsilon^2 b_1,
&&a_1\neq0,
\notag
\end{align}
where \(0<\varepsilon\ll1\). Then, in each regime,  the ABC flow possesses two periodic solutions on \(\TT^3\)
bifurcating from the corresponding degenerate family on the integrable axis.

More precisely, let
\[
\theta_+=\frac{\pi}{2},
\qquad
\theta_-=-\frac{\pi}{2}.
\] 
In the three regimes, we introduce the transverse rescaled variables
\[
z=\varepsilon Z,\qquad
x=\varepsilon X,\qquad
y=\varepsilon Y,
\]
respectively.
The first averaged vector fields in the three regimes are as follows
\begin{equation}
\label{eq:averaged-fields-main}
\begin{aligned}
\bar{\vF}_a(x,Z)
&=
\left(Z,\frac{b_1}{a}\cos x\right)^{\mathsf T},
\\
\bar{\vF}_b(y,X)
&=
\left(X,\frac{c_1}{b}\cos y\right)^{\mathsf T},
\\
\bar{\vF}_c(z,Y)
&=
\left(Y,\frac{a_1}{c}\cos z\right)^{\mathsf T}.
\end{aligned}
\end{equation}
In the corresponding   variables
$
(x,Z),\ (y,X)$, and $(z,Y)$, 
their simple zeros are
\[
P^a_\pm=(\theta_\pm,0),
\qquad
P^b_\pm=(\theta_\pm,0),
\qquad
P^c_\pm=(\theta_\pm,0).
\]

These zeros give rise to periodic solutions of the ABC flow~\eqref{eq:abc}.
We denote the corresponding periodic orbits by
\begin{equation}
\begin{aligned}
\Gamma^a_{\pm,\varepsilon}
&=
\{\vx^a_\pm(t,\varepsilon):t\in\mathbb R\},\\
\Gamma^b_{\pm,\varepsilon}
&=
\{\vx^b_\pm(t,\varepsilon):t\in\mathbb R\},\\
\Gamma^c_{\pm,\varepsilon}
&=
\{\vx^c_\pm(t,\varepsilon):t\in\mathbb R\},
\end{aligned}
\end{equation}
in regimes \(\mathrm{(A)}\), \(\mathrm{(B)}\), and \(\mathrm{(C)}\),
respectively.

In regime \(\mathrm{(A)}\), the initial points of the periodic solutions
can be chosen so that
\[
x^a_\pm(0,\varepsilon)
=
\theta_\pm+\mathcal O(\varepsilon),
\quad
y^a_\pm(0,\varepsilon)
=
0,
\quad
z^a_\pm(0,\varepsilon)
=
\mathcal O(\varepsilon^2),
\]
and their periods satisfy
\[
T^a_{\pm,\varepsilon}
=
\frac{2\pi}{|a|}+\mathcal O(\varepsilon^2).
\]

In regime \(\mathrm{(B)}\), the initial points of the periodic solutions
can be chosen so that
\[
x^b_\pm(0,\varepsilon)
=
\mathcal O(\varepsilon^2),
\quad
y^b_\pm(0,\varepsilon)
=
\theta_\pm+\mathcal O(\varepsilon),
\quad
z^b_\pm(0,\varepsilon)
=
0,
\]
and their periods satisfy
\[
T^b_{\pm,\varepsilon}
=
\frac{2\pi}{|b|}+\mathcal O(\varepsilon^2).
\]

In regime \(\mathrm{(C)}\), the initial points of the periodic solutions
can be chosen so that
\[
x^c_\pm(0,\varepsilon)
=
0,
\quad
y^c_\pm(0,\varepsilon)
=
\mathcal O(\varepsilon^2),
\quad
z^c_\pm(0,\varepsilon)
=
\theta_\pm+\mathcal O(\varepsilon),
\]
and their periods satisfy
\[
T^c_{\pm,\varepsilon}
=
\frac{2\pi}{|c|}+\mathcal O(\varepsilon^2).
\]\end{theorem}
\subsection{Proof of the existence theorem}
\label{subsec:proof-existence-theorem}

We prove Theorem~\ref{thm:main-cyclic}. The complete averaging calculation is
carried out in regime \(\mathrm{(A)}\), where the coefficient \(a\) is fixed
and nonzero. The regimes \(\mathrm{(B)}\) and \(\mathrm{(C)}\) are obtained
from this case by cyclic permutation of the variables and the parameters~\eqref{eq:sym}.

Assume
\begin{equation*}
\label{eq:A-params}
a\neq0,
\qquad
b=\varepsilon^2 b_1,
\qquad
c=\varepsilon^2 c_1,
\qquad
b_1\neq0.
\end{equation*}
Then system~\eqref{eq:abc} takes the form
\begin{equation}
\label{eq:A-system}
\begin{cases}
\dot x=a\sin z+\varepsilon^2 c_1\cos y,\\
\dot y=\varepsilon^2 b_1\sin x+a\cos z,\\
\dot z=\varepsilon^2 c_1\sin y+\varepsilon^2 b_1\cos x.
\end{cases}
\end{equation}
It is evident that for $\varepsilon=0$ we obtain the unperturbed system
\begin{equation*}
\label{eq:A-unperturbed}
\dot x=a\sin z,
\qquad
\dot y=a\cos z,
\qquad
\dot z=0.
\end{equation*}
Clearly, the two-dimensional torus $z=0$ is invariant. Indeed, on $z=0$ the
unperturbed system reduces to
\[
\dot x=0,
\qquad
\dot y=a,
\qquad
\dot z=0.
\]
This lifted solution is not periodic in \(\mathbb R^3\), because the
coordinate \(y(t)\) is linear in time. However, its projection onto
\(\mathbb T^3\) is periodic. More precisely, for
\[
T_0=\frac{2\pi}{|a|},
\]
one has
\[
x(t+T_0)=x(t),
\qquad
z(t+T_0)=z(t),
\]
and
\[
y(t+T_0)=y(t)+2\pi\,\operatorname{sgn}(a).
\]
Hence
\[
\bigl(x(t+T_0),y(t+T_0),z(t+T_0)\bigr)
\sim
\bigl(x(t),y(t),z(t)\bigr)
\quad \text{on } \mathbb T^3.
\]
Thus the invariant torus \(z=0\) is foliated by closed orbits on
\(\mathbb T^3\), each of them having minimal period $T_0$.

\subsubsection{Blow-up and angular time}
Since we seek periodic solutions that remain close to $z=0$, we introduce the
blow-up
$
z=\eps Z.
$
Using
\[
\sin(\eps Z)=\eps Z+\cO(\eps^3),
\qquad
\cos(\eps Z)=1-\frac{\eps^2Z^2}{2}+\cO(\eps^4),
\]
we obtain
\begin{equation*}
\begin{cases}
\displaystyle
\dot x
=
\eps aZ+\eps^2 c_1\cos y+\cO(\eps^3),
\\[2mm]
\displaystyle
\dot y
=
a+\eps^2\left(b_1\sin x-\frac{a}{2}Z^2\right)
+\cO(\eps^4),
\\[2mm]
\displaystyle
\dot Z
=
\eps\left(c_1\sin y+b_1\cos x\right).
\end{cases}
\end{equation*}
Since \(a\neq0\), for all sufficiently small \(\varepsilon>0\), the quantity
\(\dot y\) remains nonzero in a neighborhood of the torus \(z=0\). Hence
\(y\) can be used as a new independent variable. Dividing \(\dot x\) and
\(\dot Z\) by \(\dot y\),  and eliminating the time variable, we obtain
\begin{equation}
\label{eq:A-reduced}
\begin{cases}
\displaystyle
\frac{\rmd x}{\rmd y}
=
\eps Z
+\eps^2\frac{c_1}{a}\cos y
+\cO(\eps^3),
\\[2mm]
\displaystyle
\frac{\rmd Z}{\rmd y}
=
\frac{\eps}{a}
\left(c_1\sin y+b_1\cos x\right)
+\cO(\eps^3).
\end{cases}
\end{equation}
In particular, to first order in $\eps$ we have the standard averaging form~\eqref{eq:averaging-standard}, which is
\begin{equation}
\label{eq:A-standard}
\frac{\rmd}{\rmd y}
\begin{pmatrix}x\\ Z\end{pmatrix}
=
\eps \vF_a(y,x,Z)+\cO(\eps^2),
\end{equation}
where
\begin{equation*}
\label{eq:A-G}
\vF_a(y,x,Z)=
\begin{pmatrix}
Z\\[1mm]
\dfrac{1}{a}\left(c_1\sin y+b_1\cos x\right)
\end{pmatrix}.
\end{equation*}

\subsubsection{The averaged system}

According to the notation of Section~\ref{sec:preliminaries}, in regime
\(\mathrm{(A)}\) we have
\[
s=y,\qquad T=2\pi,\qquad \vX=(x,Z)^{\mathsf T}.
\]
Thus the first averaged vector field~\eqref{eq:averaged-function} is
\begin{equation}
\label{eq:A-average-def}
\bar{\vF}_a(x,Z)
=
\frac{1}{2\pi}
\int_0^{2\pi}
\vF_a(y,x,Z)\,\rmd y.
\end{equation}
Since
\[
\frac{1}{2\pi}\int_0^{2\pi}\sin y\,\rmd y=0,
\]
we obtain
\begin{equation}
\label{eq:A-average}
\bar{\vF}_a(x,Z)
=
\begin{pmatrix}
Z\\[1mm]
\dfrac{b_1}{a}\cos x
\end{pmatrix}.
\end{equation}
The equation \(\bar{\vF}_a(x,Z)=0\) has exactly two zeros on
\(\TT\times\RR\). In the interval $[-\pi,\pi]$, these solutions are
\begin{equation}
\label{eq:A-zeros}
P^a_\pm=(\theta_\pm,0),
\qquad
\theta_+=\frac{\pi}{2},
\qquad
\theta_-=-\frac{\pi}{2}.
\end{equation}
The Jacobian matrix of $\bar{\vF}_a(x,Z)$ is as follows
\begin{equation}
\label{eq:A-Jac}
D_{\vX}\bar{\vF}_a(x,Z)
=
\begin{pmatrix}
0&1\\[1mm]
-\dfrac{b_1}{a}\sin x&0
\end{pmatrix}.
\end{equation}
Therefore
\begin{equation*}
\label{eq:zeros}
\det D_{\vX}\bar{\vF}_a(P^a_+)
=
\frac{b_1}{a}\neq0,
\quad
\det D_{\vX}\bar{\vF}_a(P^a_-)
=
-\frac{b_1}{a}\neq0.
\end{equation*}
Hence, both zeros~\eqref{eq:A-zeros} are simple. By
Theorem~\ref{thm:first-order-averaging}, they give rise to two periodic
solutions of the perturbed ABC flow.
Applying the inverse blow-up transformation \(z=\varepsilon Z\) and returning
to the original time variable, we obtain the leading-order localization of
these periodic solutions in the original variables. 
For the reduced system, the periodic solutions bifurcating from \(P^a_\pm\)
satisfy
\[
x^a_\pm(y,\varepsilon)
=
\theta_\pm+\mathcal O(\varepsilon),
\qquad
Z^a_\pm(y,\varepsilon)
=
\mathcal O(\varepsilon).
\]
Consequently, in the original variables,
\[
z^a_\pm(y,\varepsilon)
=
\varepsilon Z^a_\pm(y,\varepsilon)
=
\mathcal O(\varepsilon^2).
\]

The condition \(y(0)=0\) is a phase normalization. Since the ABC flow is
autonomous, a time shift of a periodic solution gives the same periodic orbit.
This phase normalization is independent of the choice of the symmetric
fundamental interval \([-\pi,\pi]\). It merely selects the representative
of the periodic orbit crossing the section \(y=0\pmod{2\pi}\) at
\(t=0\). With this choice, for \(a>0\), the corresponding
periodic solutions have the leading-order form
\begin{equation}
\label{eq:A-leading-solutions}
\begin{cases}
\displaystyle
x^a_\pm(t,\varepsilon)
=
\theta_\pm+\mathcal O(\varepsilon),
\\[1mm]
\displaystyle
y^a_\pm(t,\varepsilon)
=
at+\mathcal O(\varepsilon^2),
\\[1mm]
\displaystyle
z^a_\pm(t,\varepsilon)
=
\mathcal O(\varepsilon^2).
\end{cases}
\end{equation}
In particular, we have
\[
(x^a_\pm(0,\varepsilon),y^a_\pm(0,\varepsilon),z^a_\pm(0,\varepsilon))
=
\left(
\theta_\pm+\mathcal O(\varepsilon),
0,
\mathcal O(\varepsilon^2)
\right).
\]
Thus the two periodic solutions are localized near \(x=\pi/2\) and
\(x=-\pi/2\), respectively, and remain at distance \(\mathcal O(\varepsilon^2)\)
from the unperturbed invariant torus \(z=0\).

We denote these two periodic solutions by
\begin{equation}
\Gamma^a_{\pm,\varepsilon}
=
\bigl\{
(x^a_\pm(t,\varepsilon),
y^a_\pm(t,\varepsilon),
z^a_\pm(t,\varepsilon)):
t\in\mathbb R
\bigr\}.
\end{equation}
Moreover, as \(\varepsilon\to0\), the periodic orbit
\(\Gamma^a_{\pm,\varepsilon}\) converges to the unperturbed periodic orbit
\[
\Gamma^a_{\pm,0}
=
\bigl\{
(\theta_\pm,at,0):t\in\mathbb R
\bigr\}
\subset\TT^3.
\]
To recover the period with respect to the original time \(t\), we use
\[
\frac{\rmd t}{\rmd y}=\frac{1}{\dot y}.
\]
Along the periodic solutions obtained above we have
\[
\dot y
=
a\cos\bigl(\varepsilon Z^a_\pm(y,\varepsilon)\bigr)
+
\varepsilon^2b_1\sin x^a_\pm(y,\varepsilon)
=
a+\cO(\varepsilon^2)
\]
uniformly in \(y\). Hence, for sufficiently small \(\varepsilon>0\),
\(\dot y\) does not vanish and has the same sign as~\(a\).

For \(a>0\) the positive return time corresponding to one winding
in the \(y\)-direction is therefore
\[
\begin{aligned}
T^a_{\pm,\varepsilon}
&=
\int_0^{2\pi}
\frac{dy}{
a\cos\bigl(\varepsilon Z^a_\pm(y,\varepsilon)\bigr)
+
\varepsilon^2 b_1
\sin x^a_\pm(y,\varepsilon)
}
\\[1mm]
&=
\frac{2\pi}{a}
+
O(\varepsilon^2).
\end{aligned}
\]



\subsection{The remaining cases by cyclic permutation}
\label{subsec:remaining-cases}

The remaining two cases follow directly from the \(a\)-axis analysis by the
cyclic symmetry~\eqref{eq:sym} of the ABC flow. Applying this transformation
once gives the perturbative regime near the \(b\)-axis, while applying it
twice gives the regime near the \(c\)-axis.

For the \(a\)-axis, the averaging construction is based on the invariant
torus \(z=0\), the angular variable \(y\), and the transverse rescaling
\(z=\varepsilon Z\). Under the cyclic symmetry~\eqref{eq:sym}, these
quantities transform according to
\[
z=0 \longmapsto x=0 \longmapsto y=0,
\qquad
y \longmapsto z \longmapsto x,
\]
and
\[
z=\varepsilon Z
\longmapsto
x=\varepsilon X
\longmapsto
y=\varepsilon Y.
\]
Thus the calculation carried out for the \(a\)-axis applies directly to the
other two cases.

In all three regimes, the first averaged vector field has the same form
\[
(u,V)'=(V,\mu\cos u),
\]
with the corresponding ratio of the perturbing coefficient to the dominant
one. The three averaged vector fields and their zeros are summarized in
Table~\ref{tab:cyclic-averaged-fields}.

\begin{table}[t]
\centering
\renewcommand{\arraystretch}{1.65}
\begin{tabular}{c c c c}
\toprule
Case
& Variables
& First averaged vector field
& Simple zeros
\\
\midrule
\(\mathrm{(A)}\)
&
\((x,Z)\)
&
\(\displaystyle
\bar{\vF}_a(x,Z)=
\left(
Z,\frac{b_1}{a}\cos x
\right)
\)
&
\(\displaystyle
P^a_\pm=(\theta_\pm,0)
\)
\\[3mm]
\(\mathrm{(B)}\)
&
\((y,X)\)
&
\(\displaystyle
\bar{\vF}_b(y,X)=
\left(
X,\frac{c_1}{b}\cos y
\right)
\)
&
\(\displaystyle
P^b_\pm=(\theta_\pm,0)
\)
\\[3mm]
\(\mathrm{(C)}\)
&
\((z,Y)\)
&
\(\displaystyle
\bar{\vF}_c(z,Y)=
\left(
Y,\frac{a_1}{c}\cos z
\right)
\)
&
\(\displaystyle
P^c_\pm=(\theta_\pm,0)
\)
\\
\bottomrule
\end{tabular}
\caption{First averaged vector fields and their simple zeros in the three
cyclic perturbative regimes.}
\label{tab:cyclic-averaged-fields}
\end{table}

Under the assumptions of regimes \(\mathrm{(A)}\)--\(\mathrm{(C)}\), all
zeros listed in Table~\ref{tab:cyclic-averaged-fields} are simple.
Therefore, Theorem~\ref{thm:first-order-averaging} gives two periodic
solutions of the ABC flow in each perturbative regime. This completes the
proof of Theorem~\ref{thm:main-cyclic}.

\subsection{Phase portrait of the averaged system and numerical illustration of the bifurcating periodic orbits}
\label{subsec:phase-portrait-averaged-system}
The averaged systems obtained in the proof of
Theorem~\ref{thm:main-cyclic} also determine the leading-order type of the
bifurcating periodic solutions. We first classify their equilibria and then
compare the resulting phase portrait with the Poincar\'e dynamics of the
full ABC flow.

\begin{theorem}
\label{thm:abc-stability-compact}
Under the assumptions of Theorem~\ref{thm:main-cyclic}, the eigenvalues of
the Jacobian matrices of the corresponding averaged vector fields at their
simple zeros are given as follows:
\begin{equation}
\label{eq:abc-eigenvalues}
\begin{aligned}
\mathrm{(A)}\quad
P^a_+:\quad
\lambda_{1,2}
&=
\pm\sqrt{-\frac{b_1}{a}},
&\quad
P^a_-:\quad
\lambda_{1,2}
&=
\pm\sqrt{\frac{b_1}{a}},
\\[1mm]
\mathrm{(B)}\quad
P^b_+:\quad
\lambda_{1,2}
&=
\pm\sqrt{-\frac{c_1}{b}},
&\quad
P^b_-:\quad
\lambda_{1,2}
&=
\pm\sqrt{\frac{c_1}{b}},
\\[1mm]
\mathrm{(C)}\quad
P^c_+:\quad
\lambda_{1,2}
&=
\pm\sqrt{-\frac{a_1}{c}},
&\quad
P^c_-:\quad
\lambda_{1,2}
&=
\pm\sqrt{\frac{a_1}{c}}.
\end{aligned}
\end{equation}
Consequently, the type of each equilibrium of the averaged system is
determined by the sign of the corresponding ratio
\[
\frac{b_1}{a},
\qquad
\frac{c_1}{b},
\qquad
\frac{a_1}{c}.
\]
More precisely, if the corresponding ratio is positive, then the
\(+\)-equilibrium is of center type and the \(-\)-equilibrium is a hyperbolic
saddle. If the ratio is negative, the roles are reversed.
\end{theorem}
\begin{proof}
We give the proof for regime \(\mathrm{(A)}\). The Jacobian matrix of the
averaged vector field \(\bar{\vF}_a\) is given by~\eqref{eq:A-Jac}, and its
characteristic polynomial is
\begin{equation}
\label{eq:A-charpoly}
p_a(\lambda;x)
=
\lambda^2+\frac{b_1}{a}\sin x.
\end{equation}
At the two simple zeros~\eqref{eq:A-zeros}, this gives
\[
p_a(\lambda;\theta_+)
=
\lambda^2+\frac{b_1}{a},
\qquad
p_a(\lambda;\theta_-)
=
\lambda^2-\frac{b_1}{a}.
\]
Hence
\[
P^a_+:\quad
\lambda_{1,2}=\pm\sqrt{-\frac{b_1}{a}},
\qquad
P^a_-:\quad
\lambda_{1,2}=\pm\sqrt{\frac{b_1}{a}},
\]
as stated in~\eqref{eq:abc-eigenvalues}.

The averaged system
\[
x'=Z,
\qquad
Z'=\frac{b_1}{a}\cos x,
\]
is Hamiltonian in the canonical coordinates $(x,Z)$, with the Hamiltonian function
\begin{equation}
\label{eq:A-averaged-Hamiltonian}
\mathcal H_a(x,Z)
=
\frac{1}{2}Z^2-\frac{b_1}{a}\sin x.
\end{equation}
Indeed,
\[
x'
=
\frac{\partial\mathcal H_a}{\partial Z},
\qquad
Z'
=
-\frac{\partial\mathcal H_a}{\partial x}.
\]
After the shift
\[
q=x-\frac{\pi}{2},
\]
the Hamiltonian becomes
\[
\mathcal H_a(q,Z)
=
\frac{1}{2}Z^2-\frac{b_1}{a}\cos q,
\]
which is the Hamiltonian of a mathematical pendulum. Consequently, if
\(b_1/a>0\), the equilibrium \(P^a_+\) is a nonlinear center, whereas
\(P^a_-\) is a hyperbolic saddle. If \(b_1/a<0\), the roles of the two
equilibria are reversed.

The formulas in regimes \(\mathrm{(B)}\) and \(\mathrm{(C)}\) are obtained by
applying the cyclic permutation~\eqref{eq:sym}. This gives the
remaining eigenvalue relations stated in~\eqref{eq:abc-eigenvalues}. The
classification into center and saddle type follows from the corresponding
Hamiltonian structure and the sign of the relevant parameter ratio.
\end{proof}

We now illustrate the preceding linear analysis by comparing the phase
portrait of the averaged vector field with a Poincar\'e section of the
 ABC flow in Case~\textup{(A)}. We choose
\begin{equation}
\label{eq:parki2}
a=1,\quad b=\varepsilon^2 b_1,\quad c=\varepsilon^2 c_1,\quad
\varepsilon=\frac{1}{500},\quad b_1=c_1=1.
\end{equation}
For these parameter values, the first averaged vector field is
\begin{equation}\label{eq:aver}
\frac{\rmd}{\rmd y}
\begin{pmatrix}x\\ Z\end{pmatrix}
=
\epsilon
\bar{\vF}_a(x,Z)
=\epsilon
\begin{pmatrix}
Z\\
\cos x
\end{pmatrix},
\end{equation}
and its simple zeros are
\[
P^a_+
=
\left(\frac{\pi}{2},0\right),
\qquad
P^a_-
=
\left(-\frac{\pi}{2},0\right).
\] 
 The corresponding eigenvalues~\eqref{eq:abc-eigenvalues} computed at these points are as follows
\begin{equation*}
\label{eq:eigen}
\lambda_{1,2}(P^a_+)=\pm\rmi,
\qquad
\lambda_{1,2}(P^a_-)=\pm1.
\end{equation*}
Therefore, \(P^a_+\) is an elliptic equilibrium, whereas \(P^a_-\)
is a hyperbolic saddle. By the first-order averaging theorem, these two
simple zeros generate two periodic solutions of the  ABC flow for
all sufficiently small \(\varepsilon>0\).

 \begin{figure*}[t]
\centering
\subfigure[Stream plot of averaged system~\eqref{eq:aver}]{
\includegraphics[width=0.4\textwidth]{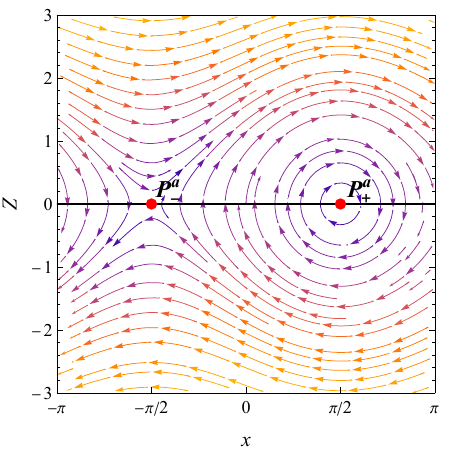}}\hspace{0.5cm}
\subfigure[Poincar\'e section of the perturbed  ABC flow~\eqref{eq:A-system}]{
\includegraphics[width=0.4\textwidth]{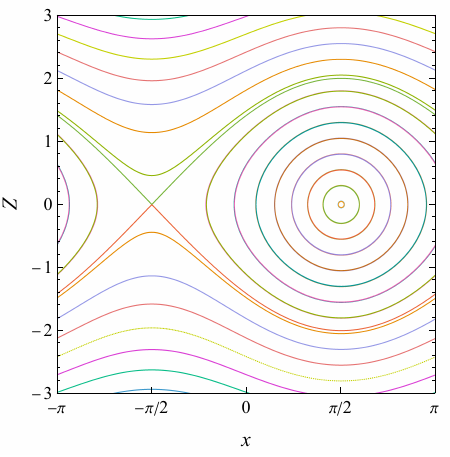}}
\caption{(Color online)
Comparison between the phase portrait of the averaged system and the
Poincar\'e section of the ABC flow in case~\textup{(A)}, for
\(a=1\), \(b_1=c_1=1\), and \(\varepsilon=1/500\).
In both panels, \(x\in[-\pi,\pi]\).
The points \(P^a_\pm\) are the equilibria of the averaged system
corresponding to the elliptic and hyperbolic periodic orbits of the
full flow. The qualitative agreement between the two panels shows
that the averaged system captures the local structure of the
Poincar\'e dynamics near the unperturbed family of periodic orbits.
}
\label{fig:Poincare-averaged-comparison}
\end{figure*}
Fig.~\ref{fig:Poincare-averaged-comparison}(a) shows the phase
portrait of the averaged system in the \((x,Z)\)-plane. The equilibrium
\(P^a_+\) is a center surrounded by closed trajectories, whereas
\(P^a_-\) is a hyperbolic saddle with the corresponding separatrix
structure, in agreement with the classification obtained above.

Fig.~\ref{fig:Poincare-averaged-comparison}(b) shows the corresponding
Poincar\'e section of the ABC flow in the blown-up coordinates
\((x,Z)\), where \(z=\varepsilon Z\). The same qualitative structure is
clearly visible: an elliptic island near \(P^a_+\) and a hyperbolic
structure near \(P^a_-\). Thus the first-order averaged system captures
the main features of the Poincar\'e dynamics near the periodic family
on \(z=0\).

This comparison also supports the numerical observations presented at
the beginning of the paper. In the unperturbed case \(b=c=0\), the
invariant torus \(z=0\) is foliated by the one-parameter family of
periodic solutions~\eqref{eq:sol}. When sufficiently small nonzero
values of \(b\) and \(c\) are introduced, this degenerate family does
not persist as a whole. Instead, two distinguished periodic solutions
bifurcate from it, corresponding to the simple zeros \(P^a_+\) and
\(P^a_-\) of the averaged system. Their elliptic and hyperbolic
characters, respectively, are clearly reflected in the Poincar\'e
section.

\begin{figure}[t]
\centering
 \includegraphics[width=0.41\textwidth]{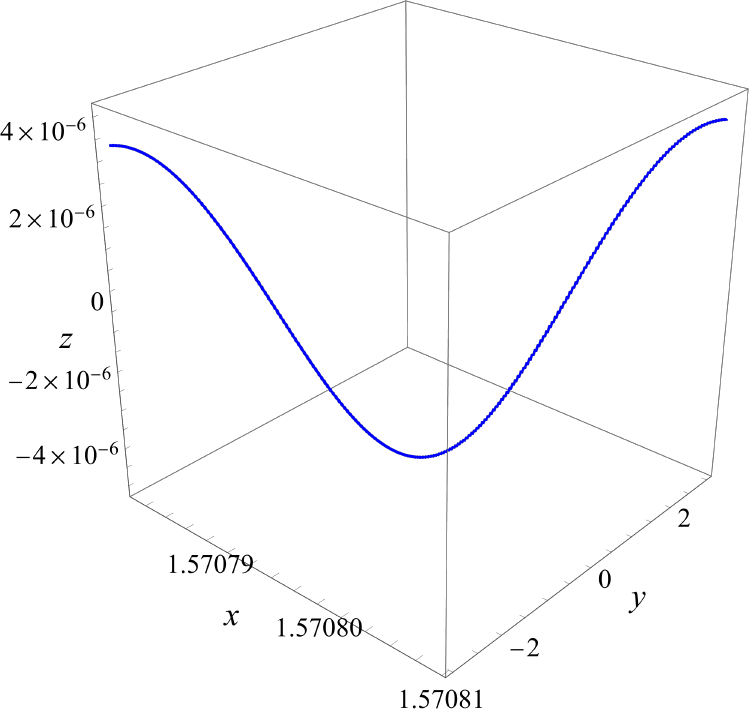}\hspace{0.5cm}
  \includegraphics[width=0.41\textwidth]{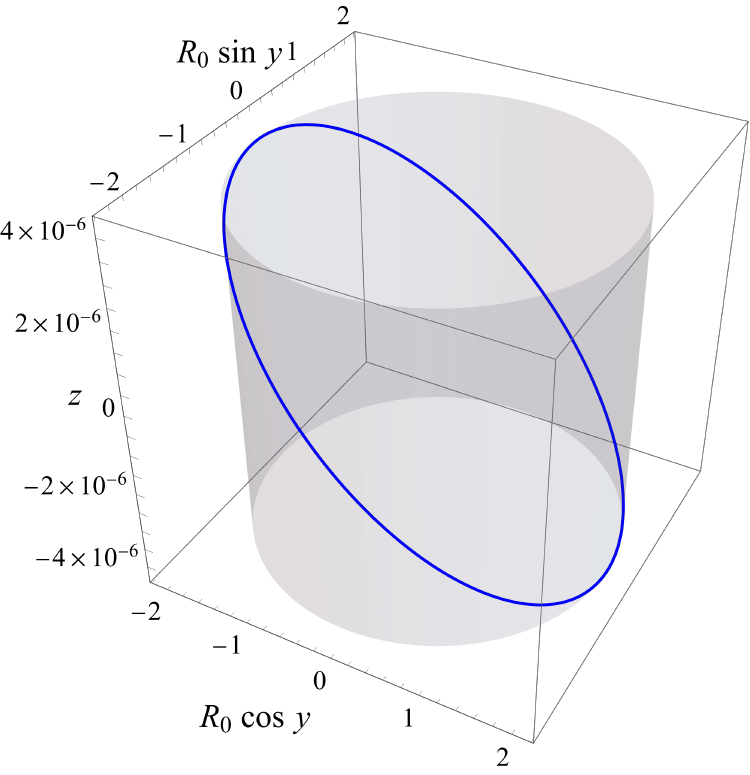}
\caption{(Color online) 
Periodic solution of the perturbed ABC flow corresponding to the simple zero
\(P^a_+=(\pi/2,0)\) of the averaged system in Case~A. The upper panel shows the lifted trajectory in the original variables
\((x,y,z)\), localized near \(x=\pi/2\) and \(z=0\). One full winding is
displayed in the symmetric interval \(y\in[-\pi,\pi]\), so that
$
y(T)-y(0)=2\pi.
$
The lower panel shows the cylindrical representation of the
\((y,z)\)-projection,
 $
(y,z)\mapsto
\bigl(R_0\cos y,R_0\sin y,z\bigr).
$}
\label{fig:abc-periodic-streamline}
\end{figure}

Fig.~\ref{fig:abc-periodic-streamline} shows a numerically computed
trajectory corresponding to the periodic orbit bifurcating from the
elliptic equilibrium \(P^a_+\). In the universal covering space, the
trajectory remains close to \(x=\pi/2\) and \(z=0\), while the lifted
angular coordinate changes by \(2\pi\) over one period. For graphical
symmetry, the phase is chosen so that the displayed winding runs from
\(y=-\pi\) to \(y=\pi\). Thus the lifted trajectory is not closed in
\(\mathbb R^3\), whereas its projection onto \(\mathbb T^3\) is a
closed periodic orbit.

The upper panel shows the lifted trajectory in the original variables
\((x,y,z)\). In the lower panel, its \((y,z)\)-projection is represented
using the cylindrical embedding
\[
(y,z)\longmapsto
\bigl(R_0\cos y,R_0\sin y,z\bigr).
\]
Since \(y\) changes by \(2\pi\) over one period, this representation
makes the periodic character of the orbit explicit.

The small transverse displacement visible in the figure is consistent
with the perturbative construction. Indeed,
\[
(x,Z)=P^a_++O(\varepsilon),
\]
and since the \(Z\)-component of \(P^a_+\) vanishes, \(Z=O(\varepsilon)\).
The blow-up \(z=\varepsilon Z\) therefore gives
$
z=O(\varepsilon^2).
$
For \(\varepsilon=1/500\), this corresponds to the scale
$
\varepsilon^2=4\times10^{-6},
$
consistent with the transverse displacement observed numerically.

\section{  \(C^{1}\) non-integrability}
\label{sec:nonint}

The proof of the following theorem combines the
Poincar\'e--Llibre--Valls criterion with the relation between the
eigenvalues of the averaged system and the characteristic multipliers
of the periodic orbits established in Theorem~\ref{thm:LZ}.
The periodic orbits obtained by averaging therefore provide the
obstruction to \(C^1\)-integrability used below.

\newpage
\begin{theorem}[Main theorem]
\label{thm:abc-nonintegrability}
Assume that one of the perturbative regimes
\textup{(A)}, \textup{(B)}, or \textup{(C)} of
Theorem~\ref{thm:main-cyclic} holds. Then, for all sufficiently small
\(\eps>0\), the characteristic multiplier \(1\) of each corresponding
periodic orbit has algebraic multiplicity one.

Consequently, the ABC flow~\eqref{eq:abc} admits no nonconstant first
integral \(H\in C^1\), defined in a neighborhood of any of these
periodic orbits, such that
\[
\nabla H(\vx)\neq0
\]
at every point \(\vx\) of the corresponding periodic orbit.
\end{theorem}
\vskip20pt
\begin{proof}[Proof of Theorem~\ref{thm:abc-nonintegrability}]
We give the proof in regime \(\mathrm{(A)}\); the other two cases follow
by cyclic symmetry. By Theorem~\ref{thm:main-cyclic}, each simple zero
\(P^a_\pm\) gives rise, for all sufficiently small \(\eps>0\), to a
periodic orbit \(\Gamma^a_{\pm,\eps}\) of the ABC flow. By
Theorem~\ref{thm:abc-stability-compact}, the corresponding eigenvalues
of the averaged Jacobian are
\[
\lambda_{1,2}(P^a_+)
=
\pm\sqrt{-\frac{b_1}{a}},
\qquad
\lambda_{1,2}(P^a_-)
=
\pm\sqrt{\frac{b_1}{a}}.
\]
Since \(a b_1\neq0\), all these eigenvalues are nonzero.

Using~\eqref{eq:full-multiplier-averaging-abc}, the two nontrivial
characteristic multipliers of either periodic orbit satisfy
\[
\mu_{j+1}(\eps)
=
\exp\!\left(
2\pi\eps\lambda_j(P^a_\pm)
\right)
+
\cO(\eps^2),
\qquad j=1,2,
\]
after a possible relabeling. Hence
\[
\mu_{j+1}(\eps)-1
=
2\pi\eps\lambda_j(P^a_\pm)+\cO(\eps^2),
\qquad j=1,2.
\]
Since \(\lambda_j(P^a_\pm)\neq0\), it follows that
\[
\mu_2(\eps)\neq1,
\qquad
\mu_3(\eps)\neq1
\]
for all sufficiently small \(\eps>0\). The remaining characteristic
multiplier is the trivial multiplier \(\mu_1(\eps)=1\), and therefore
it has algebraic multiplicity one.

Theorem~\ref{thm:PLV} now implies that the ABC flow admits no
nonconstant \(C^1\) first integral \(H\), defined in a neighborhood
of either \(\Gamma^a_{\pm,\eps}\), such that
\[
\nabla H(\vx)\neq0
\qquad
\text{for all }\vx\in\Gamma^a_{\pm,\eps}.
\]

The same argument applies in regimes \(\mathrm{(B)}\) and
\(\mathrm{(C)}\) by the cyclic symmetry~\eqref{eq:sym}. In these cases
the relevant eigenvalue pairs are
\[
\pm\sqrt{-\frac{c_1}{b}},
\quad
\pm\sqrt{\frac{c_1}{b}},
\qquad\text{and}\qquad
\pm\sqrt{-\frac{a_1}{c}},
\quad
\pm\sqrt{\frac{a_1}{c}},
\]
respectively. They are nonzero under the assumptions of
Theorem~\ref{thm:main-cyclic}, and the conclusion follows again from
Theorem~\ref{thm:PLV}.
\end{proof}

\begin{remark}\upshape
\label{rem:integrable-plane-critical-integral}
Theorem~\ref{thm:abc-nonintegrability} does not require both transverse
perturbation coefficients to be nonzero. For instance, in regime
\(\mathrm{(A)}\), the nondegeneracy condition requires \(ab_1\neq0\),
whereas \(c_1\) may vanish, since its contribution has zero average with
respect to the fast angle \(y\).

This does not contradict the integrability of the case \(c=0\). Indeed,
the ABC flow then admits the first integral
\[
H(x,y,z)=a\cos z+b\sin x.
\]
The two periodic orbits selected by averaging are, in this case, given
explicitly by
\[
x(t)\equiv\pm\frac{\pi}{2},
\qquad
z(t)\equiv0,
\qquad
y(t)=y_0+(a\pm b)t,
\]
with the signs chosen consistently. Since
\[
\nabla H(x,y,z)
=
\left(b\cos x,\,0,\,-a\sin z\right),
\]
we have
\[
\nabla H\left(\pm\frac{\pi}{2},y,0\right)=0.
\]
Thus, although a first integral exists, it is critical along the periodic
orbits selected by averaging. This is fully consistent with
Theorem~\ref{thm:abc-nonintegrability}, which excludes only \(C^1\) first
integrals that are regular along the corresponding periodic orbit.

By cyclic symmetry, the same observation applies to regimes
\(\mathrm{(B)}\) and \(\mathrm{(C)}\).
\end{remark}



 
\section{Numerical monodromy and comparison with the averaging}
\label{sec:numerical-monodromy}
 
\subsection{Numerical verification of the characteristic multipliers}
\label{subsec:numerical-multipliers}

The proof of Theorem~\ref{thm:abc-nonintegrability} uses the asymptotic
relation between the eigenvalues of the averaged Jacobian and the
characteristic multipliers of the bifurcating periodic orbits. Here we
test this relation numerically by computing the periodic orbits and their
monodromy matrices directly for the original ABC flow.

Since Theorem~\ref{thm:abc-nonintegrability} is perturbative, it does not
provide a quantitative estimate of the accuracy of the first-order
multiplier approximation at finite parameter values. We therefore compare
the characteristic multipliers predicted by averaging with those obtained
numerically from the monodromy matrices of the full ABC system.

We restrict the numerical analysis to regime \(\mathrm{(A)}\), since the
remaining two regimes follow by cyclic permutation. We use the parameter
values specified in~\eqref{eq:parki2} and compute the corresponding
periodic orbits numerically. The associated averaged vector field has the
two simple zeros \(P^a_+\) and \(P^a_-\), which generate the periodic
families \(\Gamma^a_{+,\eps}\) and \(\Gamma^a_{-,\eps}\), respectively.

All numerical computations were performed in \textit{Mathematica 15} using up to \(30\)-digit working precision. The original ABC system and the associated
variational equations were integrated with \texttt{NDSolveValue}, with
\texttt{AccuracyGoal} and \texttt{PrecisionGoal} both set to \(11\) and
\texttt{MaxSteps} set to \texttt{Infinity}. The shooting equations were
solved with \texttt{FindRoot}, using a maximum of \(200\) iterations and
accepting only solutions whose closing residual was smaller than
\(10^{-9}\). The monodromy matrix was obtained by integrating the full
\(3\times3\) fundamental matrix simultaneously with the periodic orbit.
Its accuracy was monitored through the closing error, Liouville's identity
\(\det M=1\), the autonomous-flow relation
\(M\vv(\vx_0)=\vv(\vx_0)\), and the reciprocal-pair condition
\(\mu_2\mu_3=1\).

By Theorem~\ref{thm:LZ}, the nontrivial characteristic multipliers of
the bifurcating periodic orbits are determined, to first order in
\(\eps\), by the eigenvalues of the corresponding averaged Jacobian.
For the parameter values~\eqref{eq:parki2}, this gives
\begin{equation}
\label{eq:pred}
\begin{aligned}
\Gamma^a_{+,\eps}:&\qquad
\mu_1=1,\qquad
\mu_{2,3}
=
\exp(\pm 2\pi\rmi\eps)+\cO(\eps^2),
\\
\Gamma^a_{-,\eps}:&\qquad
\mu_1=1,\qquad
\mu_{2,3}
=
\exp(\pm 2\pi\eps)+\cO(\eps^2).
\end{aligned}
\end{equation}
For \(\eps=1/500\), the corresponding first-order predictions are
\[
\mu_{2,3}
=
0.999921\pm0.012566\,\rmi
\]
for the elliptic branch, and
\[
\mu_2=1.012646,
\qquad
\mu_3=0.987512
\]
for the hyperbolic branch.

To compute the corresponding periodic orbits of the original ABC flow,
we use the leading-order locations determined by the simple zeros of the
averaged system as initial guesses:
\[
\vx_{\mathrm g}^{\,+}
=
\left(\frac{\pi}{2},0,0\right),
\qquad
\vx_{\mathrm g}^{\,-}
=
\left(-\frac{\pi}{2},0,0\right).
\]
These points determine the limiting locations of the two periodic
orbits but satisfy the periodicity conditions only to leading order.
We therefore refine both the initial conditions and the periods by a
shooting procedure.

We fix the phase by taking
$
\vx_0=(x_0,0,z_0)
$
on the section \(y=0\pmod{2\pi}\). The unknowns \(x_0\), \(z_0\), and \(T\) are determined from the lifted
periodicity conditions
\begin{equation}
\label{eq:shooting-A}
\begin{aligned}
x(T;x_0,0,z_0)-x_0&=0,\\
y(T;x_0,0,z_0)-2\pi&=0,\\
z(T;x_0,0,z_0)-z_0&=0.
\end{aligned}
\end{equation}
Here the equations are solved for the lifted variables in
\(\mathbb R^3\). The second condition describes one complete winding in
the positive \(y\)-direction, while the first and third require the
remaining coordinates to return to their initial values. Thus a
solution of~\eqref{eq:shooting-A} determines a periodic orbit on
\(\TT^3\).

The section \(y=0\pmod{2\pi}\) is transverse to the periodic orbits
considered here. Indeed,
\[
\dot y=b\sin x+a\cos z,
\]
and at the limiting orbits \(x=\pm\pi/2\), \(z=0\), we have
\[
\dot y=a\pm b.
\]
For the parameters~\eqref{eq:parki2}, \(a=1\) and
\(b=\eps^2b_1\), so \(\dot y>0\) for sufficiently small \(\eps>0\).
The choice of the positive shift \(2\pi\) in~\eqref{eq:shooting-A} is
therefore consistent with the orientation of the flow.

Once the periodic orbit and its period have been determined, we
integrate the variational equations
\begin{equation}
\label{eq:abc-variational-numerical}
\dot{\vPhi}(t)
=
D_{\vx}\vv\bigl(\vx(t)\bigr)\vPhi(t),
\qquad
\vPhi(0)=\vI_3.
\end{equation}
together with the ABC system over one period. The monodromy matrix is
then
\[
M=\vPhi(T),
\]
and its eigenvalues \(\mu_1,\mu_2,\mu_3\) are the characteristic
multipliers of the periodic orbit.

Although the shooting equations are formulated in \(\RR^3\), the
matrix \(\vPhi(T)\) also represents the monodromy matrix of the
corresponding periodic orbit on \(\TT^3\). Indeed, the identification
\[
(x,y+2\pi,z)\sim(x,y,z)
\]
is induced by a translation, whose differential is the identity.
Moreover, the ABC vector field and its Jacobian are \(2\pi\)-periodic
in each spatial variable.

The numerical computation was checked using several independent
consistency conditions. Since the ABC vector field is divergence-free,
$
\nabla\cdot\vv=0,
$
Liouville's formula gives
\begin{equation}
\label{eq:abc-liouville-check}
\det M
=
\exp\left(
\int_0^T
\nabla\cdot\vv\bigl(\vx(t)\bigr)\,\rmd t
\right)
=
1.
\end{equation}
Moreover, since the ABC system is autonomous, the tangent vector to a
periodic orbit is an eigenvector of the monodromy matrix corresponding
to the trivial multiplier \(1\), i.e.,
\begin{equation*}
\label{eq:abc-trivial-multiplier-check}
M\vv(\vx_0)=\vv(\vx_0).
\end{equation*}
It then follows from~\eqref{eq:abc-liouville-check} that the two
remaining characteristic multipliers must satisfy
\[
\mu_2\mu_3=1.
\]

To monitor the accuracy of the numerical shooting procedure and the integration
of the variational equations, we use the residuals
\[
E_{\mathrm{cl}}
=
\left\|
\varphi_T(\vx_0)-\vx_0-(0,2\pi,0)
\right\|,
\qquad
E_{\det}
=
\left|\det M-1\right|,
\]
and
\[
E_{\mathrm{tan}}
=
\left\|
M\vv(\vx_0)-\vv(\vx_0)
\right\|,
\qquad
E_{\mathrm{rec}}
=
\left|\mu_2\mu_3-1\right|,
\]
where \(\varphi_T(\vx_0)\) denotes the time-\(T\) flow map of the ABC
system in \(\RR^3\), and \(\|\cdot\|\) denotes the Euclidean norm.
The first residual measures the closing error of the periodic orbit,
while the remaining three measure the deviations from the determinant,
tangent-vector, and reciprocal-pair identities, respectively.

The resulting residuals are summarized in
Table~\ref{tab:abc-numerical-checks}. The closing, determinant, and
reciprocal-pair residuals are of order \(10^{-15}\), while the
tangent-vector residual is of order \(10^{-12}\), confirming the
numerical consistency of both the shooting and monodromy computations.

\begin{table}[t]
\centering
\small \setlength{\tabcolsep}{1.8pt}
\begin{tabular}{ccccc}
\hline
Family
&
\(E_{\mathrm{cl}}\)
&
\(E_{\det}\)
&
\(E_{\mathrm{tan}}\)
&
\(E_{\mathrm{rec}}\)
\\
\hline
\(\Gamma^a_{+,\eps}\)
&
\(9.16\times10^{-16}\)
&
\(8.99\times10^{-15}\)
&
\(1.50\times10^{-12}\)
&
\(1.11\times10^{-15}\)
\\
\(\Gamma^a_{-,\eps}\)
&
\(1.90\times10^{-15}\)
&
\(2.22\times10^{-16}\)
&
\(1.49\times10^{-12}\)
&
\(4.44\times10^{-15}\)
\\
\hline
\end{tabular}
\caption{Numerical consistency checks for the shooting procedure and
the monodromy-matrix computation. The quantities
\(E_{\mathrm{cl}}\), \(E_{\det}\), \(E_{\mathrm{tan}}\), and
\(E_{\mathrm{rec}}\) are defined in the text.}
\label{tab:abc-numerical-checks}
\end{table}

The characteristic multipliers obtained from the numerical monodromy
matrices are compared with the first-order averaging
predictions~\eqref{eq:pred} in
Table~\ref{tab:abc-floquet-comparison}.

\begin{table*}[t]
\centering
\small \setlength{\tabcolsep}{8pt}
\begin{tabular}{ccccc}
\hline
Periodic family
&
Multiplier
&
Analytical prediction
&
Numerical monodromy
&
Absolute error
\\
\hline

&
\(\mu_1\)
&
\(1\)
&
\(1.000000\)
&
\(7.33\times10^{-15}\)
\\
\(\Gamma^a_{+,\eps}\)
&
\(\mu_2\)
&
\(0.999921+0.012566\,\rmi\)
&
\(0.999921+0.012566\,\rmi\)
&
\(5.03\times10^{-8}\)
\\
&
\(\mu_3\)
&
\(0.999921-0.012566\,\rmi\)
&
\(0.999921-0.012566\,\rmi\)
&
\(5.03\times10^{-8}\)
\\
\hline

&
\(\mu_1\)
&
\(1\)
&
\(1.000000\)
&
\(4.44\times10^{-15}\)
\\ \(\Gamma^a_{-,\eps}\)
&
\(\mu_2\)
&
\(1.012646\)
&
\(1.012646\)
&
\(5.09\times10^{-8}\)
\\
&
\(\mu_3\)
&
\(0.987512\)
&
\(0.987512\)
&
\(4.96\times10^{-8}\)
\\
\hline
\end{tabular}
\caption{Comparison between the characteristic multipliers predicted by
first-order averaging and those computed numerically from the monodromy
matrices of the periodic families
\(\Gamma^a_{+,\eps}\) and \(\Gamma^a_{-,\eps}\).
The absolute error is defined by
\(\lvert\mu_j^{\mathrm{num}}-\mu_j^{\mathrm{av}}\rvert\).
The multipliers are rounded to six decimal places, whereas the errors
were computed from the full numerical values.}
\label{tab:abc-floquet-comparison}
\end{table*}

For the hyperbolic family \(\Gamma^a_{-,\eps}\), both averaging and the
numerical monodromy computation give two real reciprocal multipliers,
one inside and one outside the unit circle. The corresponding periodic
orbit is therefore hyperbolic.

For the elliptic family \(\Gamma^a_{+,\eps}\), the two nontrivial
multipliers form a nonreal complex-conjugate pair. Since the ABC flow
is volume-preserving and the trivial multiplier is equal to \(1\),
their product is equal to \(1\). Hence the two nontrivial multipliers
have unit modulus, and the corresponding periodic orbit is elliptic.

The numerical multipliers agree with the first-order averaging
predictions to approximately \(5\times10^{-8}\) for both periodic
families. In particular, the computations reproduce the spectral
structure used in Theorem~\ref{thm:abc-nonintegrability}: one trivial
multiplier equal to \(1\) and two nontrivial multipliers different
from \(1\).

\subsection{Numerical continuation and accuracy of the averaged approximation}
\label{subsec:range-validity}

We now examine how the accuracy of the first-order approximation changes
when the periodic orbits are continued away from the integrable
\(a\)-axis. We first consider the elliptic branch
\(\Gamma^a_{+,\eps}\) along the one-parameter path
\[
(a,b,c)=\left(1,b,\eps^2\right),
\qquad
\eps=\frac{1}{500},
\]
using \(b\) as the continuation parameter.

 For every value of \(b\), the periodic orbit is determined by
solving the shooting equations~\eqref{eq:shooting-A}, with the orbit
computed at the preceding parameter value used as the initial
approximation. The corresponding variational equations are then
integrated over the numerically determined period \(T=T(b)\), yielding
the monodromy matrix and its two nontrivial characteristic multipliers~$
\mu_{2,3}^{\mathrm{num}}(b).
$

Since \(b=\eps^2b_1\) and for the present choice \(a=1\), the first-order averaged prediction for the
elliptic branch can be written directly in terms of the actual parameter
\(b\) as
\begin{equation}
\label{eq:averaged-multipliers-b}
\mu_{2,3}^{\mathrm{av}}(b)
=
\exp\left(\pm 2\pi\rmi\sqrt{b}\right).
\end{equation}

For the elliptic branch, the numerical multipliers form a
complex-conjugate pair and, up to numerical precision, satisfy
\[
\mu_3^{\mathrm{num}}(b)
=
\overline{\mu_2^{\mathrm{num}}(b)}
=
\left(\mu_2^{\mathrm{num}}(b)\right)^{-1}.
\]
The averaged multipliers satisfy the corresponding relations exactly,
\[
\mu_3^{\mathrm{av}}(b)
=
\overline{\mu_2^{\mathrm{av}}(b)}
=
\left(\mu_2^{\mathrm{av}}(b)\right)^{-1}.
\]
Thus, both pairs lie on the unit circle, up to numerical errors in the
monodromy computation, and reciprocal pairing does not amplify the
absolute discrepancy between the two approximations.
We emphasize that, along this continuation, \(b_1=b/\eps^2\) is no longer
fixed. Consequently, the computations below should be interpreted as a
numerical examination of the range over which the first-order formula
\eqref{eq:averaged-multipliers-b} remains quantitatively accurate, rather
than as an extension of the perturbative theorem to arbitrary finite
values of \(b\).

Fig.~\ref{fig:mu} compares the real and imaginary parts of the
nontrivial characteristic multipliers obtained from the numerical
monodromy matrix with the corresponding first-order averaged
predictions. Since the continuation interval spans several orders of
magnitude, the parameter \(b\) is displayed on a logarithmic scale. This
representation resolves both the perturbative regime close to the
integrable \(a\)-axis and the gradual loss of accuracy at larger values
of \(b\).

For
$
b\lesssim 10^{-2},
$
the numerical and averaged multipliers correspond well. Indeed, the real parts  and the imaginary parts   lie almost exactly on top of one another throughout this parameter interval. Thus, the first-order averaged approximation remains highly accurate not only at the single perturbative parameter
value considered in the preceding subsection, but also over a substantially larger interval of the actual ABC coefficient \(b\).

A visible discrepancy appears when \(b\) reaches values of order
\(10^{-1}\). Near \(b\approx0.1\), the numerical and averaged
multipliers begin to differ in both their real and imaginary parts; see
Fig.~\ref{fig:mu}. For larger values of \(b\), the difference
increases, although the averaged formula still gives a good qualitative
description of the multipliers.

\begin{figure}[t]
\centering
\subfigure[Real parts of the characteristic multipliers]{
\includegraphics[width=0.45\textwidth]
{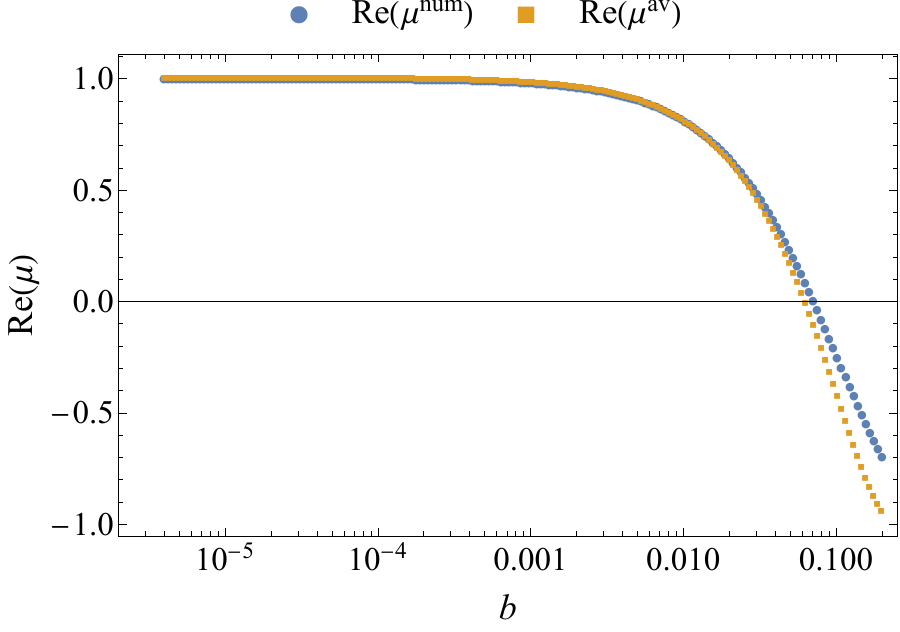}}
\subfigure[Imaginary parts of the characteristic multipliers]{
\includegraphics[width=0.45\textwidth]
{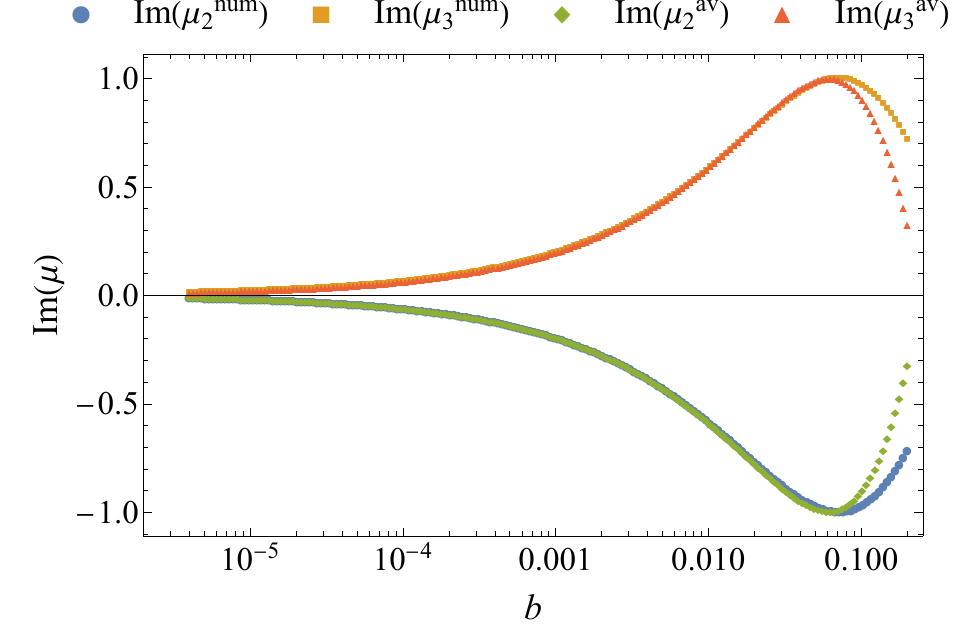}}
\caption{(Color online) 
Comparison between the nontrivial characteristic multipliers computed
numerically from the monodromy matrix and those predicted by the
first-order averaged system along the continued elliptic branch
\(\Gamma^a_{+,\eps}\). Panel~(a) shows the real parts, whereas
panel~(b) shows the imaginary parts. The continuation parameter \(b\) is
displayed on a logarithmic scale in both panels.
}
\label{fig:mu}
\end{figure}

To quantify the loss of accuracy, we introduce the maximum absolute
multiplier error
\begin{equation}
\label{eq:maximum-multiplier-error}
E(b)
=
\max_{j=2,3}
\left|
\mu_j^{\mathrm{num}}(b)
-
\mu_j^{\mathrm{av}}(b)
\right|.
\end{equation}
The dependence of \(E(b)\) on the continuation parameter is shown in
Fig.~\ref{fig:error}(a). Both axes are logarithmic, which allows the
error to be shown over several orders of magnitude. For the smallest
values of \(b\), the error is of order \(10^{-8}\), in agreement with
the comparison reported in Table~\ref{tab:abc-floquet-comparison}.
The error increases systematically with \(b\), remaining small for
\(b<10^{-2}\) and increasing noticeably for values of order \(10^{-1}\).
\begin{figure*}[t]
\centering
\subfigure[Maximum multiplier error]{
\includegraphics[width=0.459\textwidth]
{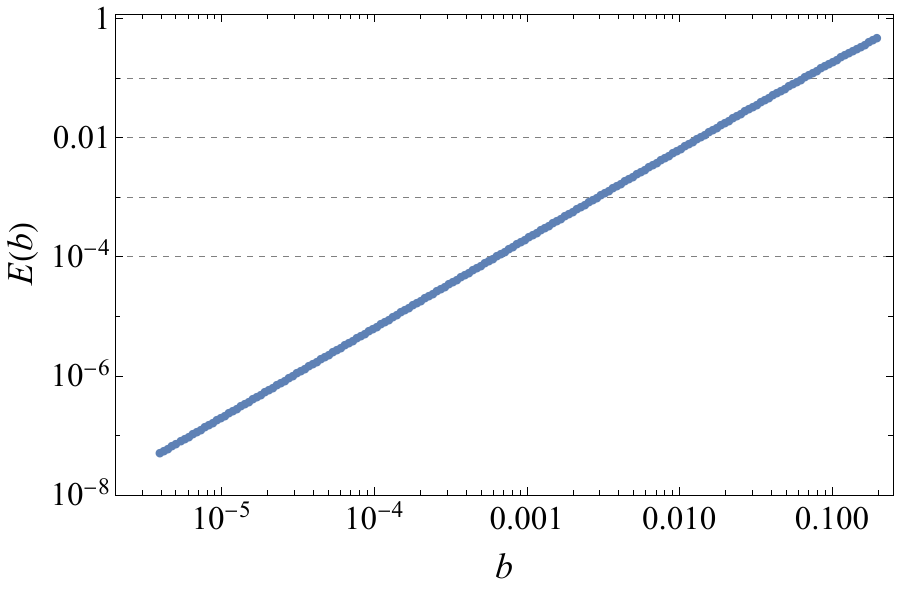}}\hspace{1cm}
\subfigure[Period of the continued periodic orbit]{
\includegraphics[width=0.45\textwidth]
{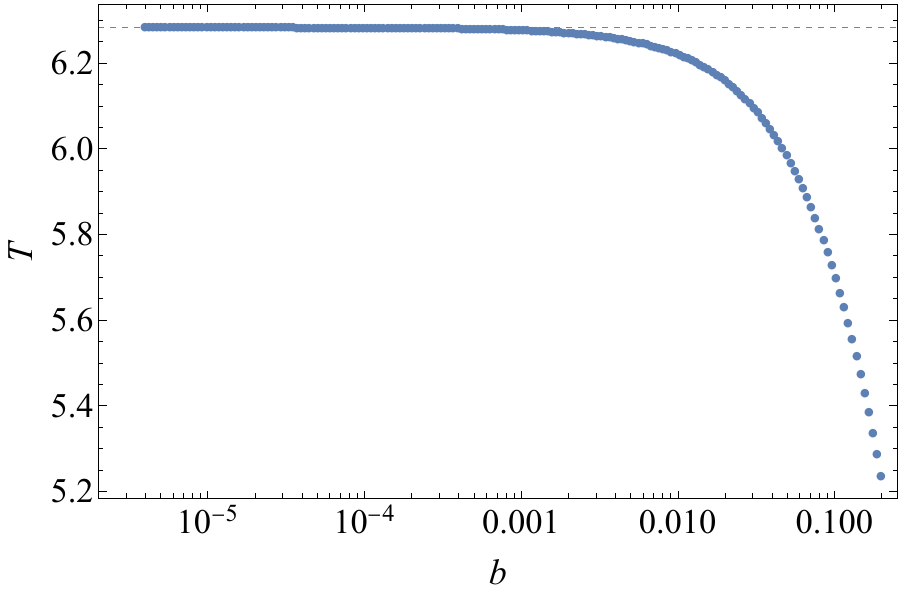}}
\caption{(Color online) 
Accuracy and continuation diagnostics for the elliptic periodic branch
\(\Gamma^a_{+,\eps}\). Panel~(a) shows the maximum multiplier error
\(E(b)\) defined in~\eqref{eq:maximum-multiplier-error} on
double-logarithmic scales. Panel~(b) shows the numerically computed period
\(T(b)\) associated with one positive winding in the \(y\)-direction;
only the horizontal axis is logarithmic in this panel. The dashed
horizontal line represents the limiting value \(T_0=2\pi\).
}
\label{fig:error}
\end{figure*}

The deterioration of the multiplier approximation is accompanied by a
deformation of the periodic orbit itself. For every value of \(b\), the
shooting procedure determines the time \(T(b)\) for which the lifted
trajectory satisfies
\[
x(T)=x(0),
\qquad
y(T)=y(0)+2\pi,
\qquad
z(T)=z(0).
\]
Therefore, \(T(b)\) is the period of the corresponding periodic orbit on
\(\TT^3\), associated with one complete winding in the positive
\(y\)-direction. In the integrable limit \(b=c=0\), it approaches the
unperturbed period \(T_0\). For the present choice \(a=1\), one has
\(T_0=2\pi\).

As shown in Fig.~\ref{fig:error}(b), where the horizontal axis is again
logarithmic, the numerically computed period remains  close to
\(2\pi\) for
$
b\lesssim 10^{-2}.
$ 
In this range, both the periodic orbit and its transverse variational
dynamics remain close to their perturbative limits, consistently with
the excellent agreement of the characteristic multipliers. For larger
values of \(b\), the period begins to decrease noticeably, and the orbit
progressively departs from the limiting periodic family.

The deviation of \(T(b)\) from \(2\pi\) indicates that the periodic
orbit moves away from its perturbative limit as \(b\) increases.
At the same time, the orbit itself is deformed, and the Jacobian matrix
\(D\vv(\vx(t))\) evaluated along it differs from its first-order
approximation. These effects, together with higher-order corrections,
account for the increasing difference between
\(\mu_{2,3}^{\mathrm{num}}\) and \(\mu_{2,3}^{\mathrm{av}}\).

For completeness, we also continue the hyperbolic branch
\(\Gamma^a_{-,\eps}\) over the same parameter interval. Throughout the
computed range, its two nontrivial characteristic multipliers remain
real, positive, and reciprocal. We denote the stable and unstable
multipliers by \(\mu_{\mathrm{s}}\) and \(\mu_{\mathrm{u}}\), respectively.
Thus,
\[
0<\mu_{\mathrm{s}}^{\mathrm{num}}(b)<1
<
\mu_{\mathrm{u}}^{\mathrm{num}}(b),
\qquad
\mu_{\mathrm{s}}^{\mathrm{num}}(b)
\mu_{\mathrm{u}}^{\mathrm{num}}(b)
=1
\]
up to numerical precision. Their first-order averaged approximations are
\[
\mu_{\mathrm{s}}^{\mathrm{av}}(b)
=
\exp\left(-2\pi\sqrt{b}\right),
\qquad
\mu_{\mathrm{u}}^{\mathrm{av}}(b)
=
\exp\left(2\pi\sqrt{b}\right).
\]
To measure the accuracy of these approximations, we define
\begin{equation*}
\label{eq:hyperbolic-multiplier-errors}
E_{\mathrm{s}}(b)
=
\left|
\mu_{\mathrm{s}}^{\mathrm{num}}(b)
-
\mu_{\mathrm{s}}^{\mathrm{av}}(b)
\right|,
\quad
E_{\mathrm{u}}(b)
=
\left|
\mu_{\mathrm{u}}^{\mathrm{num}}(b)
-
\mu_{\mathrm{u}}^{\mathrm{av}}(b)
\right|.
\end{equation*}
\begin{figure*}[t]
\centering
\subfigure[Stable and unstable characteristic multipliers]{
\includegraphics[width=0.45\textwidth]
{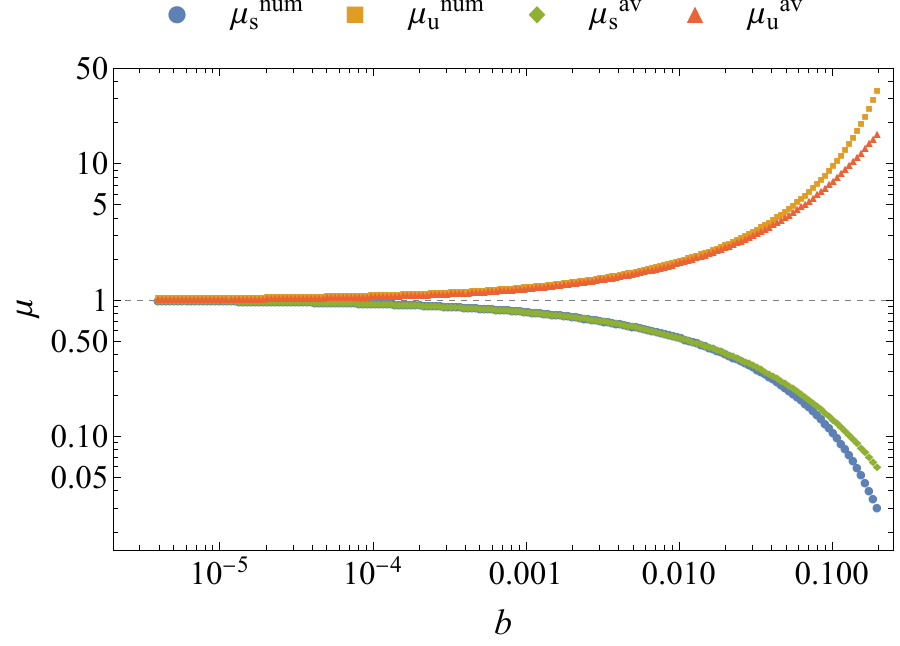}}\hspace{1cm}
\subfigure[Absolute multiplier errors]{
\includegraphics[width=0.46\textwidth]
{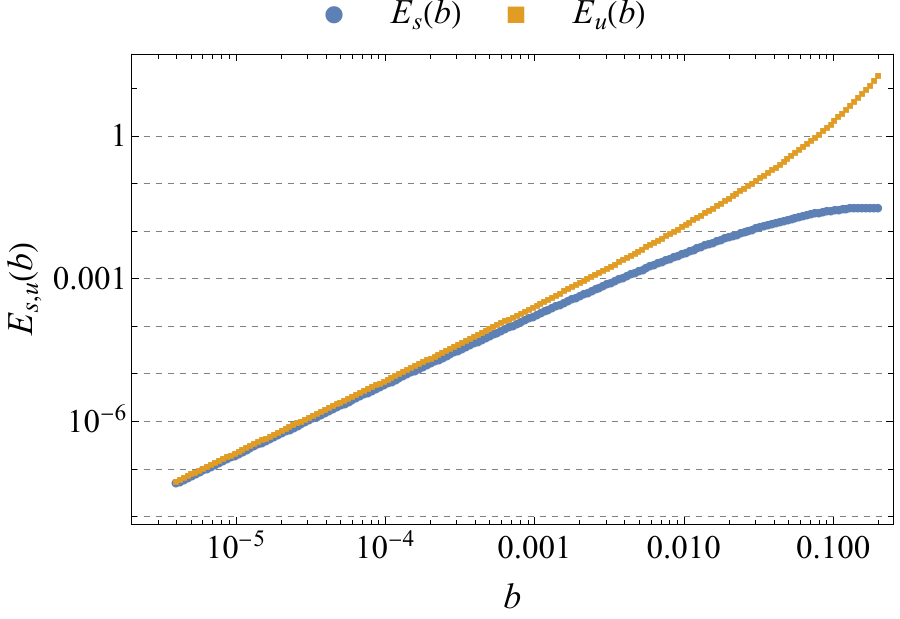}}
\caption{(Color online) 
Characteristic multipliers along the hyperbolic branch
\(\Gamma^a_{-,\eps}\). Panel~(a) compares the numerical stable and
unstable multipliers with their first-order averaged approximations on
double-logarithmic scales. Panel~(b) shows the corresponding absolute
errors, also on double-logarithmic scales.
}
\label{fig:hyperbolic-continuation}
\end{figure*}

Fig.~\ref{fig:hyperbolic-continuation}(a) compares the numerical
multipliers with their averaged approximations. For small \(b\), the
agreement is very good. As \(b\) increases, the stable multiplier
decreases from values close to \(1\), while the unstable multiplier
increases. A visible difference between the numerical and averaged
values appears for \(b\) of order \(10^{-1}\). In this range, the
numerical stable multiplier lies below its averaged approximation,
whereas the numerical unstable multiplier lies above it.

Fig.~\ref{fig:hyperbolic-continuation}(b) shows the corresponding
errors \(E_{\mathrm{s}}(b)\) and \(E_{\mathrm{u}}(b)\). Both are small
near the integrable axis and increase with \(b\). For larger values of
\(b\), the error of the unstable multiplier becomes considerably larger
than that of the stable multiplier.

This difference follows directly from the reciprocal relation. Indeed,
\[
\mu_{\mathrm{u}}^{\mathrm{num}}(b)
=
\left(\mu_{\mathrm{s}}^{\mathrm{num}}(b)\right)^{-1},
\qquad
\mu_{\mathrm{u}}^{\mathrm{av}}(b)
=
\left(\mu_{\mathrm{s}}^{\mathrm{av}}(b)\right)^{-1},
\]
and therefore
\[
E_{\mathrm{u}}(b)
=
\frac{E_{\mathrm{s}}(b)}
{
\left|
\mu_{\mathrm{s}}^{\mathrm{num}}(b)
\mu_{\mathrm{s}}^{\mathrm{av}}(b)
\right|
}.
\]
As the stable multipliers decrease, the difference between the
numerical and averaged unstable multipliers becomes larger. This
effect does not occur for the elliptic branch, since its nontrivial
multipliers remain on the unit circle. Nevertheless, the periodic
orbit remains hyperbolic throughout the investigated parameter
interval.

The continuation results show that the first-order approximation is
accurate for small \(b\). Noticeable differences between the numerical
and averaged multipliers appear when \(b\) becomes of order
\(10^{-1}\).

\section{Conclusions}

The ABC flow has long served as a fundamental model of
three-dimensional incompressible dynamics in which regular and chaotic
motions coexist. In the Poincar\'e sections considered here, this
coexistence becomes particularly transparent near the integrable
\(a\)-axis. When \(b\) and \(c\) are nonzero, the sections contain both
smooth invariant curves and an apparent stochastic region. As \(b\) and
\(c\) decrease, the visible chaotic layer becomes progressively
narrower and the dynamics approaches the degenerate integrable limit.
At the same time, two distinct structures become apparent: an elliptic
island formed by invariant curves surrounding one periodic orbit and a
broken-separatrix region associated with another, hyperbolic periodic
orbit. In the integrable limit, these isolated structures are replaced
by a continuous family of periodic streamlines. This geometric
transition provided the starting point for the perturbative analysis
developed in this work.

The natural question suggested by these sections is which members of
the degenerate periodic family survive when all three ABC parameters
become nonzero. First-order averaging gives a precise answer. Near each
integrable coordinate axis, one parameter is kept fixed and nonzero,
while the other two are scaled as quantities of order \(\eps^2\). After
a transverse blow-up and the introduction of the appropriate angular
variable as the new independent variable, the original problem reduces
to a two-dimensional periodic system. Its averaged vector field has two
simple equilibria, and these equilibria select two periodic solutions
from the degenerate family of the limiting flow. The cyclic symmetry of
the ABC vector field carries the same construction from one coordinate
axis to the other two, providing a unified description of all three
perturbative regimes.

The averaged systems also explain the different dynamical roles of the
two surviving periodic orbits. The eigenvalues of their linearizations
determine the leading-order behavior of the two nontrivial
characteristic multipliers of the corresponding periodic orbits of the
full ABC flow. In each regime, one orbit is elliptic and the other is
hyperbolic, with their assignment determined by the sign of the
relevant parameter ratio. Thus, the elliptic islands and the
hyperbolic structures visible in the Poincar\'e sections are not merely
numerical features: they arise from the two simple equilibria of the
averaged transverse dynamics.

This relation between averaged eigenvalues and characteristic
multipliers is also what connects the periodic-orbit construction with
the integrability problem. For sufficiently small \(\eps>0\), neither
of the two nontrivial multipliers is equal to \(1\). Consequently, the
trivial multiplier associated with the flow direction has algebraic
multiplicity one. The Poincar\'e--Llibre--Valls criterion then excludes
the existence of a nonconstant first integral \(H\in C^1\), defined in
a neighborhood of any of the constructed periodic orbits and regular
along it.

The central contribution of the present approach is therefore not only
the detection of periodic solutions near an integrable limit. The
periodic orbits required for the non-integrability argument are
constructed directly from the degenerate families rather than assumed
in advance or imported from a separate analysis. Moreover, the same
averaging calculation that proves their existence also classifies their
transverse dynamics and supplies the characteristic-multiplier
information needed for the non-integrability criterion. The breakup of
the degenerate periodic family, the emergence of elliptic and
hyperbolic periodic orbits, and the obstruction to regular \(C^1\)
first integrals are thereby placed within a single perturbative
framework.

Returning to the original ABC system closes the connection between this
perturbative description and the full dynamics. The periodic solutions
predicted by averaging provide effective initial approximations for the
shooting method, while integration of the full variational equations
gives their monodromy matrices and characteristic multipliers. For the
benchmark parameter values used in the direct comparison, the numerical
multipliers agree with the first-order predictions with absolute errors
of order \(10^{-8}\). The closing errors, Liouville-identity residuals,
tangent-vector residuals, and reciprocal-pair residuals provide
independent checks on these computations. Numerical continuation then
follows the elliptic and hyperbolic branches away from the immediate
asymptotic regime. It shows that the first-order predictions remain
quantitatively accurate over a wider parameter interval before the
discrepancies gradually become visible as the system moves farther from
the integrable axis.

The present results complement the Lyapunov-spectrum computations for a
broader range of ABC parameters reported in~\cite{SzuminskiLIT2026}.
Those computations describe the dynamics beyond the strict
perturbative regime, whereas the analysis developed here identifies the
mechanism by which particular periodic streamlines survive the
perturbation and explains how their transverse dynamics obstructs the
existence of regular \(C^1\) first integrals.

More broadly, the argument suggests a strategy for other
volume-preserving systems possessing degenerate families of periodic
orbits in an integrable limit. Poincar\'e sections may first reveal the
geometric structures that persist after perturbation; averaging can
then select and construct the corresponding periodic solutions, while
their characteristic multipliers provide a link between the surviving
periodic dynamics and the integrability problem.

\subsection*{Acknowledgments}

This work was carried out during a research stay of W.~S. at the Department of
Mathematics of the Universitat Autònoma de Barcelona, hosted by J.~L.

\subsection*{Funding}

W.~S. was supported by the Polish National Agency for Academic Exchange
(NAWA) through the Bekker Programme (grant no.~BPN/BEK/2025/1/00055).
J.~L. was partially supported by the Agencia Estatal de Investigación of Spain
(grant PID2022-136613NB-I00).

\subsection*{Author contributions}

W.~S. conceived the study, developed the analytical framework, performed the
analytical and numerical computations, prepared the figures and tables, and
wrote the original draft. J.~L. contributed to the theoretical discussion, the
interpretation of the non-integrability results, and the critical revision of
the manuscript. Both authors read and approved the final version of the
manuscript.

\appendix
\section{Proof of Theorem~\ref{thm:LZ}}

\label{app:LZ-proof}

\begin{proof}
Let
\[
\mathcal P_\eps(\vX_0)
=
\vX(T;0,\vX_0,\eps)
\]
denote the period map of
system~\eqref{eq:averaging-standard} based at \(s=0\).
Under the assumed regularity, integration of
system~\eqref{eq:averaging-standard} over one period gives
\begin{equation}
\label{eq:period-map-expansion}
\mathcal P_\eps(\vX_0)
=
\vX_0
+
\eps T\bar{\vF}(\vX_0)
+
\cO(\eps^2),
\end{equation}
where the remainder is \(\cO(\eps^2)\) in \(C^1\) on a
neighborhood of \(\vX_*\).

For \(\eps\neq0\), we define the normalized displacement map by
\[
\mathcal D(\vX_0,\eps)
=
\frac{\mathcal P_\eps(\vX_0)-\vX_0}{\eps}.
\]
Expansion~\eqref{eq:period-map-expansion} shows that this map extends
to \(\eps=0\) by setting
\[
\mathcal D(\vX_0,0)
=
T\bar{\vF}(\vX_0).
\]
Since
$
\mathcal D(\vX_*,0)=0
$
and
\[
D_{\vX_0}\mathcal D(\vX_*,0)
=
T D_{\vX}\bar{\vF}(\vX_*)
\]
is nonsingular, the Implicit-Function Theorem yields a fixed point
\(\vX_0(\eps)\) of \(\mathcal P_\eps\) satisfying
\[
\vX_0(\eps)
=
\vX_*+\cO(\eps).
\]
This fixed point determines the \(T\)-periodic solution
\(\vX(s,\eps)\). Since the right-hand side of
\eqref{eq:averaging-standard} is of order \(\eps\), we also have
\[
\vX(s,\eps)
=
\vX_*+\cO(\eps)
\]
uniformly for \(s\in[0,T]\).

The derivative of the period map at this fixed point is the monodromy
matrix of the corresponding reduced periodic solution. Therefore,
using the \(C^1\)-expansion~\eqref{eq:period-map-expansion}, we obtain
\begin{align*}
\vPhi(T,\eps)
&=
D_{\vX_0}\mathcal P_\eps(\vX_0(\eps))
\\
&=
\vI_m
+
\eps T
D_{\vX}\bar{\vF}(\vX_0(\eps))
+
\cO(\eps^2)
\\
&=
\vI_m
+
\eps T
D_{\vX}\bar{\vF}(\vX_*)
+
\cO(\eps^2)
\\
&=
\exp\!\left(
\eps T D_{\vX}\bar{\vF}(\vX_*)
\right)
+
\cO(\eps^2),
\end{align*}
which proves~\eqref{eq:LZ-monodromy}.

Now set
\[
\vK_\eps
=
\frac{\vPhi(T,\eps)-\vI_m}{\eps}.
\]
Then
\[
\vK_\eps
=
T D_{\vX}\bar{\vF}(\vX_*)
+
\cO(\eps).
\]
Let \(\kappa_j(\eps)\) be the eigenvalues of \(\vK_\eps\). Since the
eigenvalues \(\lambda_j\) of
\(D_{\vX}\bar{\vF}(\vX_*)\) are simple, standard matrix perturbation
theory shows that, after a suitable labeling,
\[
\kappa_j(\eps)
=
T\lambda_j+\cO(\eps),
\qquad
j=1,\ldots,m.
\]
Since
\[
\vPhi(T,\eps)
=
\vI_m+\eps\vK_\eps,
\]
the corresponding characteristic multipliers of the reduced periodic
system satisfy
\begin{align*}
\rho_j(\eps)
&=
1+\eps\kappa_j(\eps)
\\
&=
1+\eps T\lambda_j+\cO(\eps^2)
\\
&=
\exp(\eps T\lambda_j)+\cO(\eps^2),\qquad j=1,\ldots,m.
\end{align*}
which proves~\eqref{eq:multiplier-averaging}.
\end{proof}

\end{document}